\documentclass[a4paper,english,11pt]{scrartcl}
\usepackage{authblk}
\pdfoutput=1
\usepackage{tikz}
\usepackage[utf8]{inputenc}
\usepackage[T1]{fontenc}

\usepackage{mathtools,adjustbox}				
\usepackage{amssymb}				
\usepackage{mathrsfs, stmaryrd}		

\usepackage{amsthm}			

\usetikzlibrary{chains,shapes.multipart}
\usetikzlibrary{shapes,calc,patterns}
\usetikzlibrary{automata,positioning}

\usetikzlibrary{arrows, decorations.markings}
\definecolor{myred}{RGB}{220,43,25}
\definecolor{mygreen}{RGB}{0,146,64}
\definecolor{myblue}{RGB}{0,143,224}

\usetikzlibrary{shapes}

\usepackage{pgfplots}
\pgfplotsset{compat=newest} 

\usepgfplotslibrary{fillbetween}

\tikzstyle{vertex} = [shape=circle,draw=black]
\tikzstyle{namedVertex} = [shape=circle,draw=black]
\tikzstyle{edge} = [draw,->,thick]
\tikzstyle{labeledNodeS}=[circle, color=black!75!white, draw, inner sep = 0.1em, minimum size = 1.5em, scale=1.25]
\tikzstyle{normalEdge}=[very thick, >=stealth]

\usepackage[linesnumbered,ruled,nokwfunc,longend]{algorithm2e}
\DontPrintSemicolon

\usepackage{listings}
\usepackage{color}
\definecolor{purplekeywords}{rgb}{0.5,0,0.33}
\definecolor{greencomments}{rgb}{0,0.5,0}
\definecolor{bluestrings}{rgb}{0,0,1}
\definecolor{types}{rgb}{0.17,0.57,0.68}
\usepackage[bookmarks=false,colorlinks=true, linkcolor=blue, urlcolor=blue, citecolor=blue, breaklinks=true]{hyperref}
\usepackage{cleveref}			

\crefname{cons}{constraint}{constraints}
\Crefname{cons}{Constraint}{Constraints}
\creflabelformat{cons}{(#2#1#3)}
\crefname{claim}{claim}{claims}
\Crefname{claim}{Claim}{Claims}

\theoremstyle{definition}
\newtheorem{defn}{Definition}[section]

\theoremstyle{plain}

\newtheorem{prop}[defn]{Proposition}
\newtheorem{cor}[defn]{Corollary}
\newtheorem{lemma}[defn]{Lemma}
\newtheorem{theorem}[defn]{Theorem}

\newtheorem{claim}{Claim}
\newenvironment{proofClaim}[1][]{\ifthenelse{\equal{#1}{}}{\begin{proof}}{\begin{proof}[#1]}\renewcommand\qedsymbol{\ensuremath{\blacksquare}}}{\end{proof}}

\theoremstyle{remark}
\newtheorem{remark}[defn]{Remark}
\newtheorem{obs}[defn]{Observation}

\usepackage{array}					

\usepackage{braket}					
\usepackage{bm}						
\usepackage{nicefrac} 				

\allowdisplaybreaks					

\newcommand{\IN}{\mathbb{N}}	
\newcommand{\INo}{\IN_0}		
\newcommand{\INs}{\IN^\ast}		
\newcommand{\IQ}{\mathbb{Q}}	
\newcommand{\IR}{\mathbb{R}}
\newcommand{\R}{\mathbb{R}}	

\newcommand{\BigO}{\mathcal{O}}

\newcommand{\floor}[1]{\left\lfloor#1\right\rfloor}
\newcommand{\abs}[1]{\left|#1\right|}

\newcommand{\noqed}{\let\qed\relax}

\usepackage{comment}			
\usepackage{textcomp}			

\usepackage{rotating}

\usepackage{setspace}

\usepackage[font={small,it}]{caption} 

\usepackage{transparent}
\usepackage{color}
\usepackage{graphicx}

\usepackage[all]{xy}
\usepackage{lmodern}
\usepackage{tabto}

\usepackage[babel]{csquotes}

\usepackage{framed,enumitem}

\setlist[itemize]{leftmargin=1.4em}
\setlist{nosep}

\newcommand{\edgeLoad}[1][]{%
	\ifthenelse{\equal{#1}{}}
	{F^{\Delta}}
	{F^{\Delta}_{#1}}
}

\newcommand{\termTime}[1][]{%
	\ifthenelse{\equal{#1}{}}
	{\Theta}
	{\Theta_{\mathrm{#1}}}
}

\newcommand{\bigO}{\mathcal{O}}
\newcommand{\bigOm}{\Omega}

\newcommand{\network}{\mathcal{N}}

\newcommand{\ThreeSAT}{\ensuremath{\mathtt{3SAT}}}

\usepackage{dsfont} 
\newcommand{\CharF}[1][]{%
	\ifthenelse{\equal{#1}{}}
	{\mathds{1}}
	{\mathds{1}_{#1}}
}

\newcommand{\edgesLeaving}[1]{\delta^+_{#1}}
\newcommand{\edgesEntering}[1]{\delta^-_{#1}}

\newcommand{\rDeriv}[1]{\partial_+ #1}
\newcommand{\lDeriv}[1]{\partial_- #1}

\usepackage[left=1.1in, right=1.1in, top=1.2in, bottom=1.5in]{geometry}

\usepackage[numbers]{natbib}
\usepackage{etoolbox}
\makeatletter
\pretocmd{\NAT@citexnum}{\@ifnum{\NAT@ctype>\z@}{\let\NAT@hyper@\relax}{}}{}{}
\makeatother

\title{A Finite Time Combinatorial Algorithm for \\ Instantaneous Dynamic Equilibrium Flows}

\author{Lukas Graf}
\author{Tobias Harks\thanks{The research of the authors was funded by the Deutsche Forschungsgemeinschaft (DFG, German Research Foundation) - HA 8041/1-1 and HA 8041/4-1.}}
\affil{\small Augsburg University, Institute of Mathematics, 86135 Augsburg\\
\href{mailto:lukas.graf@math.uni-augsburg.de}{\{\texttt{lukas.graf,tobias.harks\}@math.uni-augsburg.de}}}

\begin{document}
\maketitle	
\begin{abstract}
	Instantaneous dynamic equilibrium (IDE) is a standard
	game-theoretic concept in dynamic traffic assignment 
	 in which individual flow particles 
	myopically select en route currently shortest paths towards their
	destination. We analyze IDE within the  Vickrey bottleneck model, where current travel times along a path consist of the physical travel times plus the sum of waiting times in all the queues along a path.
	Although IDE have been studied for decades, several fundamental questions regarding equilibrium
	computation and complexity are not well understood.
	In particular, all existence results and computational
	methods are based on fixed-point theorems and numerical discretization schemes and no exact finite
	time algorithm for equilibrium computation is known to date.
	As our main result we show that a natural extension algorithm
	needs only finitely many phases to converge
	leading to the first finite time combinatorial
	algorithm computing an IDE. We complement this result
	by several hardness results showing that computing
	IDE with natural properties is NP-hard.
\end{abstract}


\section{Introduction}

Flows over time or dynamic flows are an important mathematical concept in network flow problems with many real world applications such as dynamic traffic assignment, production systems and communication networks (e.g., the Internet). In such applications, flow particles that are sent over an edge require a certain amount of time to travel through each edge and when routing decisions are being made, the dynamic flow
propagation leads to later effects in other parts of the network.
A key characteristic of such applications, especially in traffic assignment, is that the network edges have a limited flow capacity which, when exceeded, leads to congestion. This phenomenon can be captured by the \emph{fluid queueing model} due to Vickrey~\cite{Vickrey69}. The model is based on a directed graph $G=(V,E)$, where every edge $e$ has an associated physical transit time $\tau_e\in \IR_+$ and a maximal rate capacity $\nu_e\in\IR_+$.  If flow enters an edge with higher rate than its capacity the excess particles start to form a queue at the edge's tail, where they wait until they can be forwarded onto the edge (cf. \Cref{fig:QueueingModel}). Thus, the total travel time experienced by a single particle traversing an edge $e$ is the sum of the time spent waiting in the queue of $e$ and the physical transit time $\tau_e$.

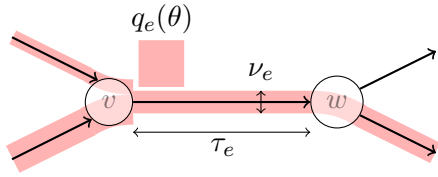
\begin{figure}\centering
	\begin{adjustbox}{max width=.5\textwidth}
		\begin{tikzpicture}
	\newcommand{\flowColor}{red!30}

	\node[namedVertex,opacity=0] (v) at (0,0) {$v$};
	\node[namedVertex,opacity=0] (w) at (3,0) {$w$};

	\draw[line width=.2cm,\flowColor] ($(v.150)+(.01,.2)$) -- +(-1,.5);
	\draw[line width=.2cm,\flowColor] ($(v.150)+(0,.2)$) to[out=330,in=180] ($(v.east)+(0,.2)$);
	\draw[line width=.4cm,\flowColor] ($(v.210)+(.01,-.1)$) -- +(-1,-.5);
	\draw[line width=.4cm,\flowColor] ($(v.210)+(0,-.1)$) to[out=30,in=180] ($(v.east)+(0,-.1)$);
	
	\draw[line width=.3cm,\flowColor] (v) -- (w);
	\draw[line width=.6cm,\flowColor] ($(v)+(.69,.2)$) -- +(0,.62);
	
	\draw[line width=.3cm,\flowColor] (w.west) to[out=0,in=150] (w.330);
	\draw[line width=.3cm,\flowColor] ($(w.330)+(-.01,0)$) -- +(1,-.5);
	
	\draw[edge,thick] (v) -- (w);
	
	\draw[edge,<-] ($(v.150)+(.13,.13)$) -- +(-1.13,.57);
	\draw[edge,<-] ($(v.210)+(.05,-.07)$) -- +(-1.05,-.53);
	
	\draw[edge] (w.30) -- +(1,.5);
	\draw[edge] (w.330) -- +(1,-.5);
	
	\node () at ($(v)+(.7,1.1)$){$q_e(\theta)$};
	
	\draw[<->] (2,-.15cm) -- +(0,.3cm);
	\node () at (2,.4) {$\nu_e$};
	
	\draw[<->] ($(v.east)+(0,-.4)$) -- ($(w.west)+(0,-.4)$);
	\node () at ($(v)!.5!(w) + (0,-.6)$) {$\tau_e$};
	
	\node[namedVertex,fill=white,fill opacity=.5] () at (v) {$v$};
	\node[namedVertex,fill=white,fill opacity=.5] () at (w) {$w$};
\end{tikzpicture}
	\end{adjustbox}
	\caption{An edge $e=vw$. As the inflow rate at node $v$ exceeds the edge's capacity, a queue forms at its tail.}\label{fig:QueueingModel}
\end{figure}

This physical flow model then needs to be enhanced with a \emph{behavioral model} prescribing the actions of flow particles.  There are two
main standard behavioral models in the  traffic assignment literature 
known as  \emph{dynamic equilibrium  (DE)} (cf.~Ran and Boyce~\cite[\S~V-VI]{Ran96}) and \emph{instantaneous dynamic equilibrium (IDE)} (\cite[\S~VII-IX]{Ran96}).
Under DE, flow particles have complete
information on the state of the network for all points in time
(including the  future evolution of all flow particles) and based on this information travel along a shortest path. The full information assumption is usually justified by assuming that the game
is played
repeatedly and a DE is then an attractor of a learning process. 
The behavioral model of IDE is based on the idea that drivers are informed in
real-time about the current traffic situation and, if beneficial,
reroute instantaneously no matter how good or bad that route
will be in hindsight. Thus, at every point in time and at every decision node, flow only enters those edges that lie on a currently shortest path towards the respective sink.
This concept assumes far less information (only the network-wide queue lengths which are continuously measured) and leads to a distributed
dynamic using only present information that is readily available via real-time information. IDE has been proposed already in the late 80's 
(cf.~Boyce, Ran and LeBlanc~\cite{BoyceRL95,RanBL93} and \citeauthor*{FrieszLTW89}~\cite{FrieszLTW89}).

A line of fairly recent works starting with Koch and Skutella~\cite{Koch11} and Cominetti, Correa and Larr\'{e}~\cite{CominettiCL15} derived
very elegant combinatorial characterizations of DE for the fluid queueing model of Vickrey. They derived a complementarity description of DE flows via so-called \emph{thin flows with resetting} which leads to an  $\alpha$-extension property stating that 
for any equilibrium up to time $\theta$, there exists  $\alpha>0$
so that the equilibrium can be extended to time $\theta+\alpha$. An extension that is maximal with respect to $\alpha$ is called a \emph{phase} in the construction of an  equilibrium and the existence of equilibria on the whole $\R_+$ then follows by a limit argument over the phases.
In the same spirit, Graf, Harks and Sering~\cite{GHS18} established 
a similar characterization for IDE flows and also derived an $\alpha$-extension property.

For both models (DE or IDE), it is an open question whether for constant inflow rates and a finite time horizon, a \emph{finite number of phases} suffices to construct  
an equilibrium, see \cite{CominettiCL15,Koch11} and~\cite{GHS18}.
This problem remains even unresolved for single-source single-sink series-parallel graphs as explicitly mentioned by Kaiser~\cite{Kaiser20}. 
Proving  finiteness of the number of phases would imply an exact finite time
algorithm. Such an algorithm is not known to date neither
for DE nor for IDE.\footnote{Algorithms for DE or IDE computation used in the transportation
science literature are \emph{numerical}, that is, only
approximate equilibrium flows are computed given a certain numerical precision,
see the related work for a more detailed comparison.}
More generally, the computational complexity of equilibrium computation
is widely open.

\subsection{Our Contribution and Proof Techniques}
In this paper, we study IDE flows and derive algorithmic and
computational complexity results.
As our main result we settle the key question
regarding finiteness of the $\alpha$-extension algorithm.
\begin{framed}
Theorem~\ref{thm:TerminationOfGenConstructionAlg}:
For single-sink networks with piecewise constant inflow rates for a finite time horizon,
there is an $\alpha$-extension algorithm computing an IDE
after finitely many extension phases. This implies the first
finite time combinatorial exact algorithm computing IDE
within the Vickrey model.
\end{framed}
The proof of our result is based on the following ideas.
We first consider the case of \emph{acyclic}
networks and use a topological order of vertices 
in order to schedule the  extension phases in the algorithm.
The key argument for the finiteness of the number of extension phases
is that for a single node $v$ and any interval with linearly changing distance labels
of nodes closer to the sink and constant inflow rate into $v$ this flow can be redistributed 
to the outgoing edges in a finite number of phases of constant outflow rates from $v$.
We show this using the properties (derivatives) of suitable edge label functions for the outgoing edges (see the graph in  \Cref{fig:Functions-hi-dist}).
The overall finiteness of the algorithm follows by induction over the nodes and time.
We then generalize to arbitrary single-sink networks
by considering \emph{dynamically} changing topological orders
depending on the current set of active edges.
Finally, a closer inspection of the proofs also enables us to give an explicit upper bound on the number of extension steps in the order of
	\[\BigO\left(P\Big(2(\Delta+1)^{4^\Delta+1}\Big)^{2L\cdot \abs{E}\cdot \abs{V}\cdot T/\tau_{\min}^2}\right),\]
where $P$ is the number of constant phases of the network inflow rates, $\Delta$ the maximum out degree in the network, $T$ the termination time, $L$ an upper bound on the absolute value of the derivatives of the distance labels depending on the network inflow rates and the edge capacities of the given network and $\tau_{\min}$ the shortest physical transit time of all edges of the given network.

We then turn to the computational complexity of IDE flows. Our first result here is a lower bound on the output complexity of any algorithm. We construct an instance in which the unique IDE flow oscillates with a changing periodicity (see \Cref{fig:CompComplexity}).
\begin{framed}
Theorem~\ref{thm:Complexity}:
	There are instances for which the output complexity of an IDE flow is not polynomial in the encoding size of the instance, even if we are allowed to use periodicity to reduce the encoding size of the flow.
\end{framed} 
We also show that several natural decision problems about the existence of IDE flows with certain properties are NP-hard.

\begin{framed}
	Theorem~\ref{thm:NPHardness}: The following decision problems are all NP-hard:
	\begin{itemize}
		\item Given a specific edge: Is there an IDE using/not using this edge?
		\item Given some time horizon $T$: Is there an IDE that terminates before $T$?
		\item Given some $k\in \IN$: Is there an IDE with at most $k$ phases?
	\end{itemize}
\end{framed}

The proof is a reduction from \ThreeSAT, wherein for any given \ThreeSAT-formula we construct a network (see \Cref{fig:3SAT-Network-Example}) with the following properties: If the \ThreeSAT-formula is satisfiable there exists a quite simple IDE flow, where all flow particles travel on direct paths towards the sink. If, on the other hand, the \ThreeSAT-formula is unsatisfiable all IDE flows in the corresponding network lead to congestions diverting a certain amount of flow into a separate part of the network. Placing different gadgets in this part of networks then allows for the reduction to various decision problems involving IDE flows.

\subsection{Related Work}

The concept of flows over time was studied by Ford and Fulkerson~\cite{Ford62}. Shortly after, Vickrey~\cite{Vickrey69} introduced 
a game-theoretic variant using a deterministic queueing model. Since then, dynamic equilibria have been studied extensively in the transportation science literature, see Friesz et al.~\cite{FrieszLTW89}. New interest in this model was raised after Koch and Skutella~\cite{Koch11} gave a novel characterization of dynamic equilibria in terms of a family of static flows (thin flows)  which was further refined by
 Cominetti, Correa and Larr\'{e} in~\cite{CominettiCL15}. Using this new approach \citeauthor*{Sering2018}~\cite{Sering2018} considered dynamic equlibria in networks with multiple sources or multiple sinks, \citeauthor*{CCO19}~\cite{CCO19} derived a bound on the price of anarchy for dynamic equilibria and \citeauthor*{Sering2019}~\cite{Sering2019} incorporated spillbacks in the fluid queuing model. In a very recent work, Kaiser~\cite{Kaiser20} showed that the thin flows needed for the extension step in computing dynamic equilibria can be determined in polynomial time for series-parallel networks. Several of these papers (\cite{CominettiCL15,CCO19,Kaiser20,Sering2018}) also explicitly mention the problem of possible non-finiteness of the extension steps.
 
 In the traffic assignment literature, the concept of IDE
 was studied  by several papers such as  Ran and Boyce~\cite[\S~VII-IX]{Ran96}, Boyce, Ran and LeBlanc~\cite{BoyceRL95,RanBL93}, Friesz et al.~\cite{FrieszLTW89}. These works develop an optimal control-theoretic formulation and characterize instantaneous user equilibria by Pontryagin's optimality conditions.
For solving the control problem, Boyce, Ran and LeBlanc~\cite{BoyceRL95} proposed to discretize
the model resulting in finite dimensional NLP whose optimal solutions correspond
to approximative IDE. While this approach only
gives an approximative equilibrium, there are further
difficulties. 
The control-theoretic formulation is actually not compatible with the  deterministic queueing model of Vickrey.
In Boyce, Ran and LeBlanc~\cite{BoyceRL95}, a differential equation
per edge governing the cumulative edge flow (state variable) is used. The right-hand side of the
differential equation depends on the exit flow function which
is assumed to be differentiable and strictly positive for any
positive inflow. Both assumptions (positivity and differentiability)
are not satisfied for the Vickrey model. For example, flow entering
an empty edge needs a strictly positive time after which it leaves 
the edge again, thus, violating the strict positiveness of the exit
flow function. More importantly, differentiability of the exit flow function
is not guaranteed for the Vickrey queueing model.
Non-differentiability (or equivalently discontinuity w.r.t. the state variable) is a well-known obstacle in the convergence
analysis of a discretization of the Vickrey model, see for instance \citeauthor*{HAN2013}~\cite{HAN2013}.  It is a priori not clear how to obtain convergence
of a discretization scheme for an arbitrary flow over time (disregading
equilibrium properties) within the Vickrey model.
And while a recent computational study by \citeauthor*{Limit2020}~\cite{Limit2020} shows some positive results with regards to convergence for DE, \citeauthor{Otsubo08}~\cite{Otsubo08} report
``significant discrepancies'' between the continuous and a discretized
solution for the Vickrey model.  To overcome the
discontinuity issue, Han et al.~\cite{HAN2013} reformulated
the model using a PDE formulation. They
obtained a discretized model whose limit points
correspond to dynamic equilibria of the continuous model.
The algorithm itself, however, is numerical in the sense that a precision
is specified and within that precision an approximate
equilibrium is computed.
The overall discretization approach mentioned above stands in line with a class of numerical algorithms based on fixed point iterations computing approximate equilibrium flows within a certain numerical precision, see Friesz and Han~\cite{Friesz19} for a recent survey.

The long term behavior of dynamic equilibria with infinitely lasting constant inflow rate at a single source was studied by  Cominetti, Correa and Olver~\cite{CominettiCO20}. They introduced the concept of a steady state and showed that dynamic equilibria always reach a stable state provided that the network inflow rate is at most the capacity of a minimal $s$-$t$ cut. Later, \Citeauthor*{NashFlowsContUniqLongTerm} showed in~\cite{NashFlowsContUniqLongTerm} that dynamic equilibria reach such stable states even without the condition on the network inflow rate if one relaxes the definition of stable state to also encompass states wherein all queues grow linearly at a fixed rate forever.

Ismaili~\cite{ismaili2017,Ismaili18} considered a discrete version of DE and IDE,
respectively.
He investigated the computational complexity of computing best responses
for DE showing that the best-response optimization problem is not approximable, and that deciding the existence of a Nash equilibrium is complete for the second level of the polynomial hierarchy. In~\cite{Ismaili18} a sequential version of
a discrete routing game is studied and PSPACE hardness results 
for computing an optimal routing strategy are derived. For further
results regarding a discrete packet routing model,
we refer to  Cao et al.~\cite{CaoCCW17}, Scarsini et al.\cite{ScarsiniST18}, Harks et al.~\cite{HarksPSTK18}
and Hoefer et al.~\cite{HoeferMRT11}.


\section{Model and the Extension-Algorithm}\label{sec:Model}

Throughout this paper we always consider networks $\network = (G,(\nu_e)_{e \in E},(\tau_e)_{e \in E},(u_v)_{v \in V\setminus\set{t}}, t)$ given by a directed graph $G=(V,E)$, edge capacities $\nu_e \in \IQ_{>0}$, edge travel times $\tau_e \in \IQ_{>0}$, 
and a single sink node $t \in V$ which is reachable from anywhere in the graph. Every other node $v \in V\setminus\set{t}$ has a corresponding (network) inflow rate $u_v: \IR_{\geq 0} \to \IQ_{\geq 0}$ indicating for every time $\theta \in \IR_{\geq 0}$ the rate $u_v(\theta)$ at which the infinitesimal small agents enter the network at node $v$ and start traveling through the graph until they leave the network at the common sink node $t$. We will assume that these network inflow rates are right-constant step functions with bounded support and finitely many, rational jump points and denote by $P \in \INs$ the total number of jump points for all network inflow rates. 

A \emph{flow over time} in $\network$ is a tuple $f = (f^+,f^-)$ where $f^+,f^-: E \times \IR_{\geq 0} \to \IR_{\geq 0}$ are integrable functions. For any edge $e \in E$ and time $\theta \in \IR_{\geq 0}$ the value $f^+_e(\theta)$ describes the \emph{(edge) inflow rate} into $e$ at time $\theta$ and $f^-_e(\theta)$ is the \emph{(edge) outflow rate} from $e$ at time $\theta$. 
For any such flow over time $f$ we define the \emph{cumulative (edge) in- and outflow rates} $F^+$ and $F^-$ by
\begin{align*}
	F^+_e(\theta) \coloneqq \int_{0}^{\theta}f^+_e(\zeta)d\zeta \quad\text{ and }\quad F^-_e(\theta) \coloneqq \int_{0}^{\theta}f^-_e(\zeta)d\zeta, 	
\end{align*}
respectively. The queue length of edge $e$ at time $\theta$ is then defined as
\begin{align}\label{Def:QueueLength}
	q_e(\theta) \coloneqq F^+_e(\theta)-F^-_e(\theta+\tau_e). 
\end{align}

Such a flow $f$ is called a \emph{feasible flow} for the given set of inflow rates $u_v: \IR_{\geq 0} \to \IQ_{\geq 0}$, if it satisfies the following \cref{eq:FeasibleFlow-FlowConservation,eq:FeasibleFlow-FlowConservationSink,eq:FeasibleFlow-NoOutflowBeforeTau,eq:FeasibleFlow-QueueOpAtCap}.
The \emph{flow conservation constraints} are modeled for all nodes $v \neq t$ as
\begin{align}\label[cons]{eq:FeasibleFlow-FlowConservation} 
	\sum_{e \in \edgesLeaving{v}}f^+_e(\theta) - \sum_{e \in \edgesEntering{v}}f^-_e(\theta) = u_v(\theta) & \quad \text{ for all } \theta \in \IR_{\geq 0},
\end{align}
where $\edgesLeaving{v} := \set{vu \in E}$ and $\edgesEntering{v} := \set{uv \in E}$ are the sets of outgoing edges from $v$ and incoming edges into $v$, respectively. For the sink node $t$ we require
\begin{align}\label[cons]{eq:FeasibleFlow-FlowConservationSink} 
	\sum_{e \in \edgesLeaving{t}}f^+_e(\theta) - \sum_{e \in \edgesEntering{t}}f^-_e(\theta) \leq 0
\end{align}
and for all edges $e \in E$ we always assume 
\begin{align}\label[cons]{eq:FeasibleFlow-NoOutflowBeforeTau}
f_e^-(\theta) = 0 &\text{ for all } \theta < \tau_e.
\end{align}
Finally we assume that the queues operate at capacity which can be modeled by
\begin{align}	\label[cons]{eq:FeasibleFlow-QueueOpAtCap} 
f_e^-(\theta + \tau_e) = 
\begin{cases}
\nu_e, & \text{ if } q_e(\theta) > 0 \\
\min\set {f^+_e(\theta), \nu_e}, & \text{ if } q_e(\theta) \leq 0
\end{cases}  
& \quad \text{ for all }  e \in E, \theta \in \IR_{\geq 0}.
\end{align}

Following the definition in \cite{GHS18} we call a feasible flow an IDE flow if whenever a particle arrives at a node $v \neq t$, it can only ever enter an edge that is the first edge on a currently shortest $v$-$t$ path. In order to formally describe this property we first define the \emph{current} or \emph{instantaneous travel time} of an edge $e$ at $\theta$ by
\begin{align}
	c_e(\theta) \coloneqq \tau_e + \frac{q_e(\theta)}{\nu_e}.\label{Def:InstantaneousTravelTime}
\end{align}
We then define time dependent node labels $\ell_{v}(\theta)$ corresponding to current shortest path distances from $v$ to the sink $t$. For $v\in V$ and $\theta\in \R_{\geq 0}$, define
\begin{equation}
	\ell_{v}(\theta)\coloneqq
	\begin{cases} 
	0, & \text{ for } v=t\\
	\min\limits_{e=vw\in E} \{\ell_{w}(\theta)+c_{e}(\theta)\}, & \text{ else.}
	\end{cases}
\end{equation}
We say that an edge $e=vw$ is \emph{active} at time $\theta$, if $ \ell_{v}(\theta) = \ell_{w}(\theta)+c_{e}(\theta)$, denote the set of active edges by $E_\theta\subseteq E$ and call the subgraph $G[E_{\theta}]$ induced by these edges the \emph{active subgraph}.

\begin{defn}\label{Def:IDE}
	A feasible flow over time $f$ is an \emph{instantaneous dynamic equilibrium (IDE)}, if for all  $\theta\in \R_{\geq 0}$ and $e\in E$ it satisfies 
	\begin{align}\label[cons]{eq:DefIDE-OnlyUseSP}
	f_{e}^+(\theta)>0 \Rightarrow e\in E_\theta.
	\end{align}
\end{defn}

During the computation of an IDE we also need the concept of \emph{partial} flows/IDE that are only defined up to a certain point in time. First, a \emph{partial flow over time} is a tupel $(f^+,f^-)$ such that for every edge $e$ we have two integrable functions $f^+_e: [0,a_e) \to \IR_{\geq 0}$ and $f^-_e: [0,a_e+\tau_e) \to \IR_{\geq 0}$ for some non-negative number $a_e$, satisfying \cref{eq:FeasibleFlow-QueueOpAtCap,eq:FeasibleFlow-NoOutflowBeforeTau} for all $\theta < a_e$. Such a flow is a \emph{feasible (partial) flow up to $\hat{\theta}$ at node $v$} if 
\begin{itemize}
	\item the edge outflow rates for all edges leading towards $v$ are defined at least up to time $\hat{\theta}$, i.e. $a_e + \tau_e \geq \hat{\theta}$ for all $e \in \edgesEntering{v}$,
	\item the edge inflow rates for all edges leaving $v$ are defined up to time $\hat{\theta}$, i.e. $a_e = \hat{\theta}$ for all $e \in \edgesLeaving{v}$ and
	\item \cref{eq:FeasibleFlow-FlowConservation} or \cref{eq:FeasibleFlow-FlowConservationSink}, respectively, holds at $v$ for all $\theta < \hat{\theta}$.
\end{itemize}
A partial flow is a \emph{feasible (partial) flow up to time $\hat{\theta}$}, if it is a feasible partial flow up to time $\hat{\theta}$ at every node. We call such a flow a \emph{partial IDE up to time $\hat{\theta}$}, if additionally \cref{eq:DefIDE-OnlyUseSP} holds for all edges and all times before $\hat{\theta}$. If the given network is acyclic, we can even speak of a \emph{partial IDE up time $\hat{\theta}$ at some node $v$}, denoting a feasible flow up to time $\hat{\theta}$ at $v$, at least up to time $\hat{\theta}$ for all nodes lying on some path from $v$ to the sink $t$ and satisfying \cref{eq:DefIDE-OnlyUseSP} for all $e \in \edgesLeaving{v}$ and $\theta < \hat{\theta}$.

Note that, while the edge inflow rates of a feasible partial flow up to $\hat{\theta}$ are defined only on $[0,\hat{\theta})$, this already determines the queue length functions and, therefore, the instantaneous edge travel times on $[0,\hat{\theta}]$. In particular, for such a flow we can speak about active edges at time $\hat{\theta}$ even though the flow itself is not yet defined at that time.

In \cite[Section 3]{GHS18} the existence of IDE flows in single-sink networks is proven by the following almost constructive argument: A partial IDE up to some time $\hat{\theta}$ can always be extended for some additional proper\footnote{We call an interval $[a,b)$ \emph{proper} if $a < b$.}{} time interval on a node by node basis (starting with the nodes closest to the sink $t$). The existence of IDE for the whole $\IR_{\geq 0}$ then follows by a limit argument. This leads to a natural algorithm for computing IDE flows in single-sink networks, which we make explicit here as \Cref{alg:ConstructionOld}, wherein $b_v^-$ denotes the \emph{gross node inflow rate} at node $v$ defined by setting
\[b_v^-(\theta) \coloneqq \sum_{e \in \edgesEntering{v}}f^-_e(\theta) + u_v(\theta)\]
for all $v \in V \setminus \set{t}$ and $\theta \in [\hat{\theta},\hat{\theta}+\tau_{\min})$, where $\tau_{\min}\coloneqq \min\set{\tau_e | e \in E} > 0$. 

\begin{algorithm}\caption{IDE-Construction Algorithm from \cite{GHS18}}\label{alg:ConstructionOld}
	\KwIn{A single-sink network $\network$ with piecewise constant network inflow rates}
	
	\KwOut{An IDE flow $f$ in $\network$}
	
	Let $f$ be the zero flow and $\theta \leftarrow 0$
	
	\While{not all flow particles have reached the sink $t$}{
		\SetNoFillComment
		\tcc{$f$ is a partial IDE up to time $\theta$}
		
		Let $t=v_1 < v_2 < \dots < v_n$ be a topological order w.r.t. $G[E_{\theta}]$
		
		\For{$i = 2, \dots, n$}{
			Compute $b_{v_i}^-(\theta)$ and determine a constant distribution of this inflow to edges in $\edgesLeaving{v_i}$ such that the used edges remain active for some proper interval\label{alg:ConstructionOld:NodeWiseExtension}
		}
		Determine the largest $\alpha \geq 0$ such that all $b_v^-$ are constant on $(\theta,\theta+\alpha)$ and the set of active edges does not change
		
		Extend $f$ up to time $\theta+\alpha$ with constant edge inflow rates and set $\theta \leftarrow \theta+\alpha$
	}
\end{algorithm} 

For the extension at a single node $v$ in \cref{alg:ConstructionOld:NodeWiseExtension} we can use a solution to the following convex optimization problem, which can be determined in polynomial time using a simple water filling procedure (see \Cref{appendix} for more details):
\begin{align}
	\min\quad&\sum_{e=vw \in \edgesLeaving{v}\cap E_{\theta}} \int_0^{x_e}\frac{g_e(z)}{\nu_e}+\rDeriv{\ell_{w}}(\theta) dz \label{OPT:bvThetaK}\tag{OPT-$b_v^-(\theta)$}\\
	\text{s.t.}\quad&  \sum_{e \in \edgesLeaving{v}\cap E_{\theta}} x_e=b_v^-(\theta), \quad x_e\geq 0 \text{ for all } e \in \edgesLeaving{v}\cap E_{\theta},\notag
\end{align}
where $g_e$ denotes the right side derivative of the queue length function $q_e$ depending on the inflow rate into $e$, i.e. $g_e(z) \coloneqq  z-\nu_e,$ if $q_e(\theta) > 0$ and $g_e(z) \coloneqq  \max\set{z-\nu_e,0}$, otherwise. The right side derivatives $\rDeriv{\ell_{w}}(\theta)$ exist because we only need them for nodes $w$ closer to the sink with respect to the current topological order. And for those we already determined (constant) edge inflow rates for all outgoing edges for some additional proper time interval beginning with $\theta$. The integrand of the objective function is, thus, the right derivative of the shortest instantaneous travel time towards the sink when entering edge $vw$ at time $\theta$ and assuming a constant inflow rate of $z$ into this edge starting at time $\theta$. Using this observation one can show (cf. \cite[Lemma 3.1]{GHS18}) that any solution to \eqref{OPT:bvThetaK} corresponds to a flow distribution to active edges so that for every edge $e = vw \in \edgesLeaving{v}\cap E_{\theta}$ the following condition is satisfied
\begin{align}\label{eq:rDerivConditionForExtension}\begin{split}
		f^+_e(\theta)&>0\quad\implies\quad \rDeriv{\ell_v}(\theta)=\rDeriv{c_e}(\theta)+\rDeriv{\ell_{w}}(\theta) \\
		f^+_e(\theta)&=0\quad\implies\quad \rDeriv{\ell_v}(\theta)\leq \rDeriv{c_e}(\theta)+\rDeriv{\ell_{w}}(\theta).
\end{split}\end{align}
Because the network inflow rates as well as all already constructed edge inflow rates are piecewise constant and the node label functions as well as the queue length functions are continuous, \eqref{eq:rDerivConditionForExtension} ensures that the used edges will remain active for some proper time interval.

It is, however, not obvious whether a finite number of such extension phases suffices to construct an IDE flow for all of $\IR_{\geq 0}$. Since IDE flows always have a finite termination time in single-sink networks (\cite[Theorem 4.6]{GHS18}) it is at least enough to extend the flow for some finite time horizon (in \cite{GH20PoA} we even provide a way to explicitly compute such a time horizon). This leaves the possibility of continuously decreasing lengths of the extension phases as possible reason for \Cref{alg:ConstructionOld} not to terminate within finite time, e.g. some sequence of extension phases of lengths $\alpha_1, \alpha_2, \dots$ such that $\sum_{i=1}^{\infty}\alpha_i$ converges to some point strictly before the IDE's termination time (see \Cref{rem:SmallExtensionPhases} for an example where we can in fact achieve arbitrarily small extension phases). Thus, the question of whether IDE flows can actually be computed was left as an open question in \cite{GHS18}. A first partial answer was found in \cite{Kraus2019}, where finite termination was shown for graphs obtained by series composition of parallel edges. In the following section we give a full answer by showing that the $\alpha$-extension algorithm terminates for all single-sink networks.


\section{Finite IDE-Construction Algorithm}\label{sec:AlgTermination}

In this chapter we will show that IDE flows can be constructed in finite time using \Cref{alg:ConstructionOld} or slight variations thereof. We will first show this only for acyclic networks since there we can use a single constant order of the nodes for the whole construction. Building on that, we will then prove the general case by showing that we can always compute IDE flows while changing the node order only finitely many times.

\subsection{Acyclic Networks}\label{sec:AlgTerminationAcyclicNetworks}

For each extension step, \Cref{alg:ConstructionOld} takes a partial IDE and determines a network-wide constant extension of all edge in- and outflow rates. This flow distribution then continues until an event (change of gross node inflow rate or change of the set of active edges) \emph{anywhere} in the network requires a new flow split. In \cite{GHS18} such a maximal extension is called a \emph{phase} of the constructed IDE. After each phase, one then has to determine a new topological order with respect to the active subgraph at the beginning of the next phase.

For an acyclic network, we can instead use a single static topological order of the nodes with respect to the whole graph, which is then in particular a topological order with respect to any possible active subgraph. This allows us to rearrange the order of the extension steps: 
Considering the nodes according to the fixed topological order, at each node, we then already know the gross node inflow rate for the whole interval $[\theta,\theta+\tau_{\min})$ \emph{as well as} the flow distribution for all nodes closer to the sink over the same time interval. Thus, we have enough information to determine a (possibly infinite) sequence of extensions covering the whole interval $[\theta,\theta+\tau_{\min})$, where each extension is defined through constant flow distributions at this node. Within this sequence, each extension lasts until an event at the current node happens, which forces us to compute a new flow distribution. 
We call such a maximal extension using one constant flow distribution at a single node a \emph{local phase}.
The restructured version of the extension algorithm is formalized in \Cref{alg:Construction}. 

\begin{algorithm}\caption{IDE-Construction Algorithm for acyclic networks}\label{alg:Construction}
	\KwIn{An acyclic single-sink network $\network$ with piecewise constant network inflow rates}
	
	\KwOut{An IDE flow $f$ in $\network$}
	
	Choose $T \in \IQ$ large enough such that all IDE flows in $\network$ terminate before $T$
	
	Let $f$ be the zero flow, $\theta \leftarrow 0$ and $t = v_1 < \dots < v_n$ a topological order
	
	\For{$k = 0, \dots, \floor{T/\tau_{\min}}$}{
		\SetNoFillComment
		\tcc{$f$ is a partial IDE up to time $\theta = k\tau_{\min}$}
		
		\For{$i = 2, \dots, n$}{
			Compute the piecewise constant gross node inflow function $b_{v_i}^-$ for the interval $[\theta,\theta+\tau_{\min})$
			
			Distribute this inflow for the whole interval to active edges in $\edgesLeaving{v_i}$ using maximal local phases of constant flow distribution\label{alg:Construction:FlowDistribution}
		}
		$\theta \leftarrow \theta + \tau_{\min}$
	}
\end{algorithm} 

\begin{obs}
	For acyclic networks both variants of the general algorithm (\Cref{alg:ConstructionOld} and \Cref{alg:Construction}) construct the same IDE provided that they use the same tie-breaking rules. Thus, showing that one of them terminates in finite time, also proves the same for the other variant.
\end{obs}

Using the water filling procedure (\Cref{Alg:FindExtension}) we can compute an IDE compliant flow distribution with constant edge-inflow rates at a node $v_i$ for any interval wherein the inflow into node $v_i$ is constant, the labels on all the nodes $w$ with $v_iw \in \edgesLeaving{v_i}$ change linearly and the set of active edges leaving $v_i$ remains constant. Thus, it suffices to show that in \cref{alg:Construction:FlowDistribution} we can always cover the extension interval $[\theta,\theta+\tau_{\min})$ with a finite number of local phases. We will show this by induction over $k \in \INo$ and $i \in [n]$ using the following key lemma:

\begin{lemma}\label{lemma:SmallestExtension}
	Let $\network$ be a single-sink network on an acyclic graph with some fixed topological order on the nodes, $v$ some node in $\network$ and $\theta_1 < \theta_2 \leq \theta_1 + \tau_{\min}$ two points in time. If $f$ is a partial flow over time in $\network$ such that
	\begin{itemize}
		\item $f$ is a partial IDE up to time $\theta_2$ for all nodes closer to the sink $t$ than $v$ with respect to the fixed topological order,
		\item $f$ is a partial IDE up to time $\theta_1$ for all other nodes,
		\item $b_v^-$ is constant during $[\theta_1,\theta_2)$ and
		\item the labels at the nodes reachable via direct edges from $v$ are affine functions on $[\theta_1,\theta_2)$,
	\end{itemize}
	then we can extend $f$ to a partial IDE up to time $\theta_2$ at $v$ using a finite number of local phases.
\end{lemma}

\begin{figure}\centering
	\begin{minipage}[c]{0.55\textwidth}
		\begin{adjustbox}{max width=\textwidth}
			\begin{tikzpicture}[scale=2]
	\newcommand{\colOne}{green!80!blue!80}
	\newcommand{\colTwo}{blue!80}
	\newcommand{\colThree}{purple!80!blue!80}
	\newcommand{\colFour}{red!80}
	\newcommand{\colFive}{orange!80}
	\newcommand{\flowColor}{gray!50}

	\newcommand{\ellGraph}[5]{
		\draw[fill=white,white] (#1) rectangle +(1,.7);
		\draw[edge] ($(#1)+(.1,.2)$) -- +(0,.4);
		\draw[edge] ($(#1)+(.1,.2)$) -- +(.9,0);
		\node () at ($(#1)+(.15,.1)$) {\small$\theta_1$};
		\node () at ($(#1)+(.9,.1)$) {\small$\theta_2$};
		
		\draw[#2,thick] ($(#1)+(.1,#3)$) --node[above,#2]() {$\ell_{w_#5}$} ($(#1)+(.9,#4)$);
	}

	
	\node (u1) at (2,1) {};
	\node (u2) at (1.5,0) {};
	\node (u3) at (2,-1) {};

	\node (v) at (2.5,0) {};	
	
	\node (w1) at (4,1.2) {};
	\node (w2) at (4,.4) {};
	\node (w3) at (4,-.4) {};
	\node (w4) at (4,-1.2) {};
	
	\node[namedVertex,gray] (t) at (6,0) {$t$};
		
	\draw[\flowColor,line width=.6cm] ($(v)+(0,1.2)$) -- (w1.center);
	\draw[\flowColor,line width=.6cm] ($(v)+(.5,1.4)$) -- +(0,.3);
	\draw[\flowColor,line width=.092cm] ($(v)+(0,.4)$) -- (w2.center);
	\draw[\flowColor,line width=.6cm] ($(v)+(.5,.5)$) -- +(0,.1);
	\draw[\flowColor,line width=.075cm] ($(v)+(0,-.4)$) -- (w3.center);
	\draw[\flowColor,line width=.6cm] ($(v)+(.5,-.3)$) -- +(0,.15);
	\draw[\flowColor,line width=.15cm] ($(v)+(0,-1.2)$) -- (w4.center);
	\draw[\flowColor,line width=.6cm] ($(v)+(.5,-1.1)$) -- +(0,.05);
	\draw[edge,\colFour] ($(v)+(0,-1.2)$) -- (w4);	
	
	\draw[edge,gray,dashed] (w1) to[out=0,in=90] (t);
	\draw[edge,gray,dashed] (w2) to[out=0,in=150] (t);
	\draw[edge,gray,dashed] (w3) to[out=0,in=-150] (t);
	\draw[edge,gray,dashed] (w4) to[out=0,in=-90] (t);
	
	\node[namedVertex,\colOne,fill=white] (w1u) at (w1) {$w_1$};
	\node[namedVertex,\colTwo,fill=white] (w2u) at (w2) {$w_2$};
	\node[namedVertex,\colThree,fill=white] (w3u) at (w3)  {$w_3$};
	\node[namedVertex,\colFour,fill=white] (w4u) at (w4) {$w_4$};
	
	\draw[ultra thick,edge,\colOne,dotted] ($(v)+(0,1.2)$) -- (w1u);
	\draw[ultra thick,edge,\colTwo,dashdotted] ($(v)+(0,.4)$) -- (w2u);
	\draw[ultra thick,edge,\colThree,dashed] ($(v)+(0,-.4)$) -- (w3u);
	\draw[ultra thick,edge,\colFour] ($(v)+(0,-1.2)$) -- (w4u);

	\draw[\flowColor,line width=.4cm] (u2.center) -- (v);
	\draw[edge] (u2.center) -- ($(v)+(-.25,0)$);

	\draw[fill=white] (2.5,0) ellipse (.25 and 1.5);
	\node() at (v) {$v$};

	\ellGraph{4.5,.9}{\colOne,dotted}{.4}{.55}{1}
	\ellGraph{4.5,.1}{\colTwo,dashdotted}{.35}{.6}{2}
	\ellGraph{4.5,-.7}{\colThree,dashed}{.3}{.4}{3}
	\ellGraph{4.5,-1.5}{\colFour}{.27}{.23}{4}
\end{tikzpicture}
		\end{adjustbox}
	\end{minipage}\hfill
	\begin{minipage}[c]{0.44\textwidth}
		\caption{The situation in \Cref{lemma:SmallestExtension}: We have an acyclic graph with some topological order on the nodes (here from left to right) and a partial IDE up to some time $\theta_2$ for all nodes closer to the sink $t$ than $v$ and up to some earlier time $\theta_1$ for $v$ and all nodes further away than $v$ from $t$. Additionally, over the interval $[\theta_1,\theta_2)$ the edges leading into $v$ have a constant outflow rate and the nodes $w_i$ all have affine label functions $\ell_{w_i}$. The edges $vw_i$ start with some current queue lengths $q_{vw_i}(\theta_1) \geq 0$.}\label{fig:situationOfLemmaSmallestExtension}
	\end{minipage}
\end{figure}
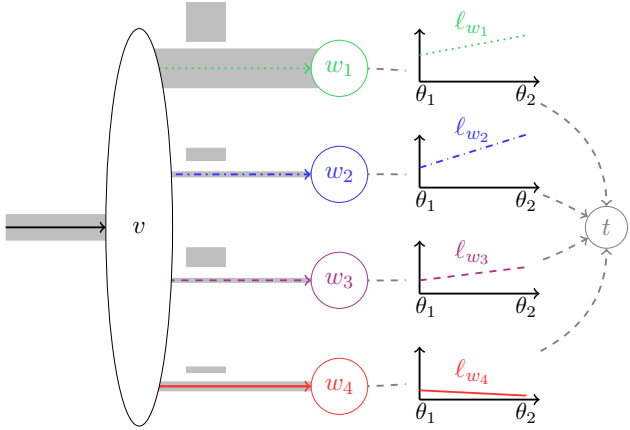

\begin{proof}
	We want to show that a finite number of maximal constant extensions of the flow at node $v$ using the water filling algorithm  is enough to extend the given flow for the whole interval $[\theta_1,\theta_2)$ at node $v$. So, let $f$ be the flow after an, a priori, infinite number of extension steps getting us to a partial IDE up to some $\hat{\theta} \in (\theta_1,\theta_2]$ at node $v$.
	
	Let $\edgesLeaving{v} = \set{vw_1, \dots, vw_p}$ be the set of outgoing edges from $v$. Then, by the lemma's assumption, the label functions $\ell_{w_i}: [\theta_1, \theta_2) \to \IR_{\geq 0}$ are affine functions and, since we extended $f$ at node $v$ up to $\hat{\theta}$, the queue length functions $q_{vw_i}$ are well defined on the interval $[\theta_1,\hat{\theta})$. Thus, for all $i \in [p]$ we can define functions
		\[h_i: [\theta_1, \hat{\theta}) \to \IR_{\geq 0}, \theta \mapsto \tau_{vw_i} + \frac{q_{vw_i}(\theta)}{\nu_{vw_i}} + \ell_{w_i}(\theta)\]
	such that $h_i(\theta)$ is the shortest  current travel time to the sink $t$ for a particle entering edge $vw_i$ at time $\theta$. Then, for any edge $vw_i \in \edgesLeaving{v}$ and any time $\theta \in [\theta_1,\hat{\theta})$ we have
	\begin{align}\label{eq:SmallestExtension:CharActiveEdge}
		vw_i \in E_{\theta} \iff h_i(\theta) = \min\set{h_j(\theta) | j \in [p]} = \ell_v(\theta).
	\end{align}
	
	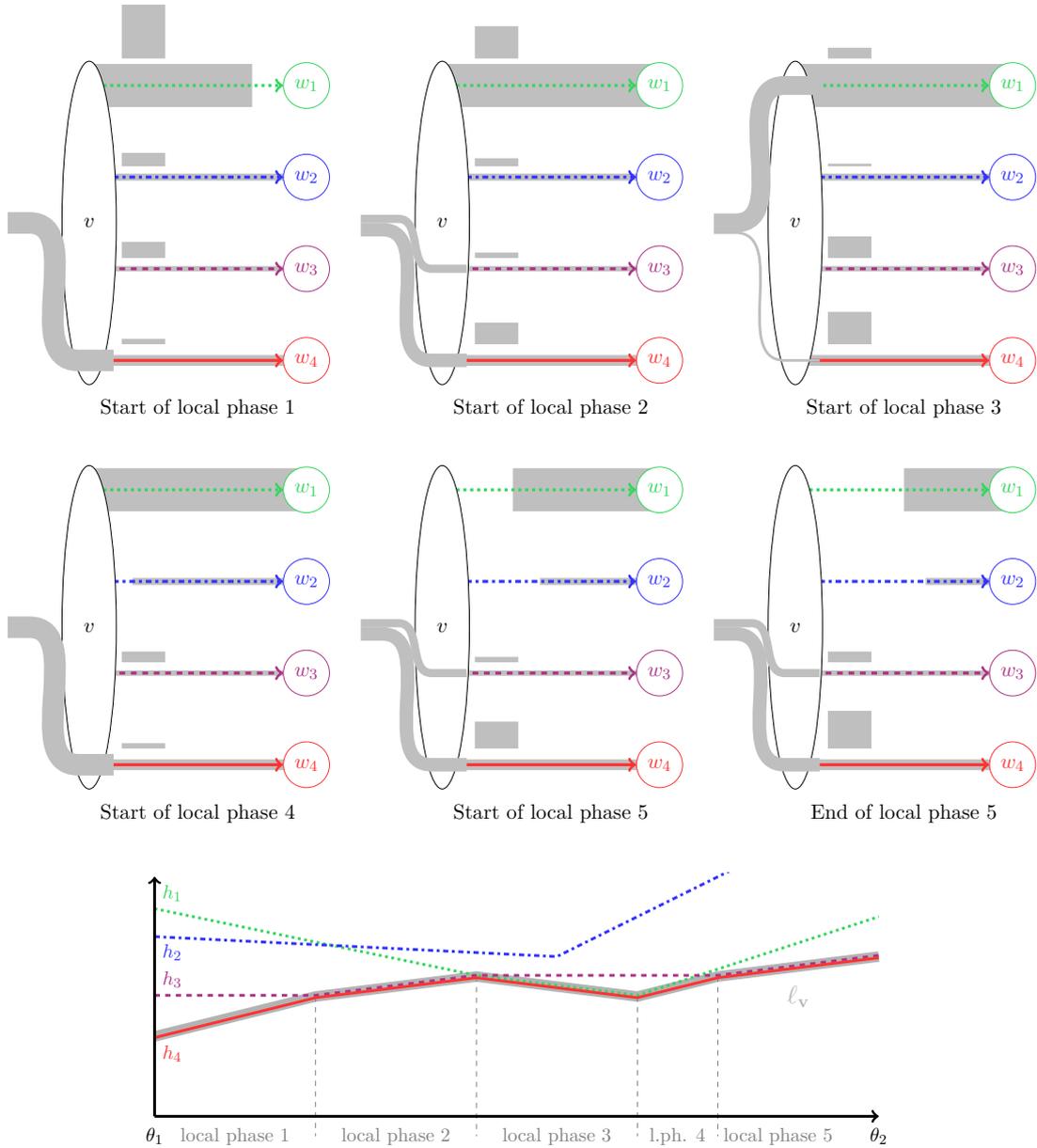
\begin{figure}\centering
		\begin{adjustbox}{max width=.95\textwidth}
			\begin{tikzpicture}
	\newcommand{\colOne}{green!80!blue!80}
	\newcommand{\colTwo}{blue!80}
	\newcommand{\colThree}{purple!80!blue!80}
	\newcommand{\colFour}{red!80}
	\newcommand{\colFive}{orange!80}
	\newcommand{\flowColor}{gray!50}
	
	\newcommand{\eOne}{5.1}
	\newcommand{\eTwo}{3.4}
	\newcommand{\eThree}{1.7}
	\newcommand{\eFour}{0}

	\newcommand{\edges}[2]{		
		\node (v1) at (0,\eOne) {};
		\node (v2) at (0,\eTwo) {};
		\node (v3) at (0,\eThree) {};
		\node (v4) at (0,\eFour) {};

		\node[namedVertex,\colOne,fill=white] (w1) at (4,\eOne) {$w_1$};
		\node[namedVertex,\colTwo,fill=white] (w2) at (4,\eTwo) {$w_2$};
		\node[namedVertex,\colThree,fill=white] (w3) at (4,\eThree) {$w_3$};
		\node[namedVertex,\colFour,fill=white] (w4) at (4,\eFour) {$w_4$};
		
		\draw[edge,ultra thick,dotted,\colOne] (v1) -- (w1);
		\draw[edge,ultra thick,dashdotted,\colTwo] (v2) -- (w2);
		\draw[edge,ultra thick,dashed,\colThree] (v3) -- (w3);
		\draw[edge,ultra thick,\colFour] (v4) -- (w4);
		
		\draw[fill=white] (0,2.55) ellipse (.5 and 3);
		\node (v) at (0,2.55) {$v$};
		
		\node[rectangle] () at (2,-.9) {\ifthenelse{\equal{#2}{end}}{End}{Start} of local phase #1};
	}

	\begin{scope}[xshift=0cm,yshift=0cm]
		\draw[\flowColor,line width=.8cm] (0,\eOne) -- (3,\eOne);
		\draw[\flowColor,line width=.8cm] (1,\eOne+.5) -- (1,\eOne+1.5);
		\draw[\flowColor,line width=.13cm] (0,\eTwo) -- (4,\eTwo);
		\draw[\flowColor,line width=.8cm] (1,\eTwo+.2) -- (1,\eTwo+.45);
		\draw[\flowColor,line width=.1cm] (0,\eThree) -- (4,\eThree);
		\draw[\flowColor,line width=.8cm] (1,\eThree+.2) -- (1,\eThree+.5);
		\draw[\flowColor,line width=.2cm] (0,\eFour) -- (4,\eFour);
		\draw[\flowColor,line width=.8cm] (1,\eFour+.3) -- (1,\eFour+.4);
		
		\edges{1}{}
		
		\draw[\flowColor,line width=.4cm] ($(v)+(-1.5,0)$) -- ($(v)+(-1.1,0)$) to[out=0,in=180] (-.2,\eFour) -- (.45,\eFour);
	\end{scope}
	
	\begin{scope}[xshift=6.5cm,yshift=0cm]
		\draw[\flowColor,line width=.8cm] (0,\eOne) -- (4,\eOne);
		\draw[\flowColor,line width=.8cm] (1,\eOne+.5) -- (1,\eOne+1.1);
		\draw[\flowColor,line width=.13cm] (0,\eTwo) -- (4,\eTwo);
		\draw[\flowColor,line width=.8cm] (1,\eTwo+.2) -- (1,\eTwo+.35);
		\draw[\flowColor,line width=.1cm] (0,\eThree) -- (4,\eThree);
		\draw[\flowColor,line width=.8cm] (1,\eThree+.2) -- (1,\eThree+.3);
		\draw[\flowColor,line width=.2cm] (0,\eFour) -- (4,\eFour);
		\draw[\flowColor,line width=.8cm] (1,\eFour+.3) -- (1,\eFour+.7);
		
		\edges{2}{}
		
		\draw[\flowColor,line width=.15cm] ($(v)+(-1.5,.075)$) -- ($(v)+(-.8,.075)$) to[out=0,in=180] (-.2,\eThree) -- (.45,\eThree);
		\draw[\flowColor,line width=.25cm] ($(v)+(-1.5,-.125)$) -- ($(v)+(-1.1,-.125)$) to[out=0,in=180] (-.2,\eFour) -- (.45,\eFour);
	\end{scope}
	
	\begin{scope}[xshift=13cm,yshift=0cm]
		\draw[\flowColor,line width=.8cm] (0,\eOne) -- (4,\eOne);
		\draw[\flowColor,line width=.8cm] (1,\eOne+.5) -- (1,\eOne+.7);
		\draw[\flowColor,line width=.13cm] (0,\eTwo) -- (4,\eTwo);
		\draw[\flowColor,line width=.8cm] (1,\eTwo+.2) -- (1,\eTwo+.25);
		\draw[\flowColor,line width=.1cm] (0,\eThree) -- (4,\eThree);
		\draw[\flowColor,line width=.8cm] (1,\eThree+.2) -- (1,\eThree+.6);
		\draw[\flowColor,line width=.2cm] (0,\eFour) -- (4,\eFour);
		\draw[\flowColor,line width=.8cm] (1,\eFour.3) -- (1,\eFour+.9);
		
		\edges{3}{}
		
		\draw[\flowColor,line width=.35cm] ($(v)+(-1.5,0)$) -- ($(v)+(-1.1,0)$) to[out=0,in=180] (-.2,\eOne) -- (.45,\eOne);
		\draw[\flowColor,line width=.05cm] ($(v)+(-1.5,-.18)$) -- ($(v)+(-1.1,-.18)$) to[out=0,in=180] (-.2,\eFour) -- (.45,\eFour);
	\end{scope}
	
	\begin{scope}[xshift=0cm,yshift=-7.5cm]
		\draw[\flowColor,line width=.8cm] (0,\eOne) -- (4,\eOne);
		\draw[\flowColor,line width=.13cm] (0.8,\eTwo) -- (4,\eTwo);
		\draw[\flowColor,line width=.1cm] (0,\eThree) -- (4,\eThree);
		\draw[\flowColor,line width=.8cm] (1,\eThree+.2) -- (1,\eThree+.4);
		\draw[\flowColor,line width=.2cm] (0,\eFour) -- (4,\eFour);
		\draw[\flowColor,line width=.8cm] (1,\eFour+.3) -- (1,\eFour+.4);
		
		\edges{4}{}
		
		\draw[\flowColor,line width=.4cm] ($(v)+(-1.5,0)$) -- ($(v)+(-1.1,0)$) to[out=0,in=180] (-.2,\eFour) -- (.45,\eFour);
	\end{scope}

	\begin{scope}[xshift=6.5cm,yshift=-7.5cm]
		\draw[\flowColor,line width=.8cm] (1.3,\eOne) -- (4,\eOne);
		\draw[\flowColor,line width=.13cm] (1.8,\eTwo) -- (4,\eTwo);
		\draw[\flowColor,line width=.1cm] (0,\eThree) -- (4,\eThree);
		\draw[\flowColor,line width=.8cm] (1,\eThree+.2) -- (1,\eThree+.3);
		\draw[\flowColor,line width=.2cm] (0,\eFour) -- (4,\eFour);
		\draw[\flowColor,line width=.8cm] (1,\eFour+.3) -- (1,\eFour+.8);
		
		\edges{5}{}
		
		\draw[\flowColor,line width=.15cm] ($(v)+(-1.5,.075)$) -- ($(v)+(-.8,.075)$) to[out=0,in=180] (-.2,\eThree) -- (.45,\eThree);
		\draw[\flowColor,line width=.25cm] ($(v)+(-1.5,-.125)$) -- ($(v)+(-1.1,-.125)$) to[out=0,in=180] (-.2,\eFour) -- (.45,\eFour);	
	\end{scope}

	\begin{scope}[xshift=13cm,yshift=-7.5cm]
		\draw[\flowColor,line width=.8cm] (2,\eOne) -- (4,\eOne);
		\draw[\flowColor,line width=.13cm] (2.4,\eTwo) -- (4,\eTwo);
		\draw[\flowColor,line width=.1cm] (0,\eThree) -- (4,\eThree);
		\draw[\flowColor,line width=.8cm] (1,\eThree+.2) -- (1,\eThree+.4);
		\draw[\flowColor,line width=.2cm] (0,\eFour) -- (4,\eFour);
		\draw[\flowColor,line width=.8cm] (1,\eFour+.3) -- (1,\eFour+1);
		
		\edges{5}{end}
		
		\draw[\flowColor,line width=.15cm] ($(v)+(-1.5,.075)$) -- ($(v)+(-.8,.075)$) to[out=0,in=180] (-.2,\eThree) -- (.45,\eThree);
		\draw[\flowColor,line width=.25cm] ($(v)+(-1.5,-.125)$) -- ($(v)+(-1.1,-.125)$) to[out=0,in=180] (-.2,\eFour) -- (.45,\eFour);	
	\end{scope}
	
\end{tikzpicture}
		\end{adjustbox}
	
		\vspace{1.5em}
		
		\begin{adjustbox}{max width=.7\textwidth}
			\begin{tikzpicture}[scale=.8]
	\newcommand{\colOne}{green!80!blue!80}
	\newcommand{\colTwo}{blue!80}
	\newcommand{\colThree}{purple!80!blue!80}
	\newcommand{\colFour}{red!80}
	\newcommand{\colFive}{orange!80}
	\newcommand{\flowColor}{gray!60}
	
	\useasboundingbox (-.5,-1) rectangle (18.5,6.1);
	\clip (-.5,-1) rectangle (18.5,6.1);
	
	\draw[\flowColor,line width=6pt] (0,2) -- (4,3) -- (8,3.5) -- (12,3) -- (14,3.5) -- (18,4);
	\node[\flowColor] () at (16,3) {\Large $\mathbf{\ell_v}$};
	
	\node[gray] () at (2,-.5) {local phase 1};
	\draw[dashed,gray] (4,3) -- (4,-.5);
	\node[gray] () at (6,-.5) {local phase 2};
	\draw[dashed,gray] (8,3.5) -- (8,-.5);
	\node[gray] () at (10,-.5) {local phase 3};
	\draw[dashed,gray] (12,3) -- (12,-.5);
	\node[gray] () at (13,-.5) {l.ph. 4};
	\draw[dashed,gray] (14,3.5) -- (14,-.5);
	\node[gray] () at (15.5,-.5) {local phase 5};

	\draw[ultra thick,dotted,\colOne] (0,5.2) node[above right] {$h_1$} -- (8,3.53) -- (12,3.03) -- (18,5);	
	\draw[ultra thick,dashdotted,\colTwo] (0,4.5) node[below right] {$h_2$} -- (10,4) -- (18,8);
	\draw[ultra thick,dashed,\colThree] (0,3.03) node[above right] {$h_3$} -- (4,3.03) -- (8,3.53) -- (14,3.53) -- (18,4.03);
	\draw[ultra thick,\colFour] (0,1.97) node[below right] {$h_4$} -- (4,2.97) -- (8,3.47) -- (12,2.97) -- (14,3.47) -- (18,3.97);

	\draw[edge,ultra thick] (0,0) -- (0,6);
	\draw[edge,ultra thick] (0,0) -- (18,0);
	
	\node () at (0,-.5) {$\theta_1$};
	\node () at (18,-.5) {$\theta_2$};
\end{tikzpicture}
		\end{adjustbox}	
		\caption{A possible flow distribution from the node $v$ in five local phases for the situation depicted in \Cref{fig:situationOfLemmaSmallestExtension}. The first six pictures show the flow split for these five local phases. The graph at the bottom shows the corresponding functions $h_i$. The bold gray line marks the graph of the function $\ell_v$. The second, third and fifth local phase all start because an edge becomes newly active (edges $vw_3$, $vw_1$ and $vw_3$ again, respectively). The fourth local phase starts because the queue on the active edge $vw_1$ runs empty. By observation~\ref{obs:TriggerNewPhase} these are the only two possible events which can trigger the beginning of a new local phase. Edge $vw_2$ is inactive for the whole time interval and -- as stated in \Cref{claim:SmallestExtension:InactiveHareConvex} -- has a convex graph. Also, note the slope changes of the functions $h_i$ and $\ell_v$ in accordance with \Cref{claim:SmallestExtension:DerivativeChanges}.}\label{fig:Functions-hi-dist}
	\end{figure}

	We start by stating two important observations and then proceed by showing two key-properties of the functions $h_i$ and $\ell_v$, which are also visualized in \Cref{fig:Functions-hi-dist}:
	
	\begin{enumerate}[label=(\roman*)]
		\item The functions $h_i$ are continuous and piece-wise linear. In particular they are differentiable almost everywhere and their left and right side derivatives $\lDeriv{h_i}$ and $\rDeriv{h_i}$, respectively, exist everywhere. The same holds for the function $\ell_v$.
		\item A new local phase begins at a time $\theta \in [\theta_1, \hat{\theta})$ if and only if at least one of the following two events occurs at time $\theta$: An edge $vw_i$ becomes newly active or the queue of an active edge $vw_i$ runs empty.\label{obs:TriggerNewPhase}
	\end{enumerate}
	
	\begin{claim}\label{claim:SmallestExtension:InactiveHareConvex}
		If an edge $vw_i$ is inactive during some interval $(a,b) \subseteq [\theta_1,\hat{\theta}]$ the graph of $h_i$ is convex on this interval.
	\end{claim}

	\begin{claim}\label{claim:SmallestExtension:DerivativeChanges}
		For any time $\theta$ define $I(\theta) \coloneqq \set{i \in [p] | h_i(\theta) = \ell_v(\theta)}$. Then, we have
		\begin{align}\label{eq:LowerBoundOnRightDerivativeOfEllV1}
			\min\set{\lDeriv{h_i}(\theta) | i \in I(\theta)} \leq \rDeriv{\ell_v}(\theta).
		\end{align}
		If, additionally, no edge becomes newly active at time $\theta$, we also have
		\begin{align}\label{eq:LowerBoundOnRightDerivativeOfEllV2}
			\lDeriv{\ell_v}(\theta) \leq \rDeriv{\ell_v}(\theta).
		\end{align}
	\end{claim}

	\begin{proofClaim}[Proof of \Cref{claim:SmallestExtension:InactiveHareConvex}]
		By the lemma's assumption $\ell_{w_i}$ is linear on the whole interval. For an inactive edge $vw_i$ its queue length function consists of at most two linear sections: One where the queue depletes at a constant rate of $-\nu_e$ and one where it remains constant $0$. Thus, $h_i$ is convex as sum of two convex functions for any interval, where $vw_i$ is inactive.
	\end{proofClaim}
	
	\begin{proofClaim}[Proof of \Cref{claim:SmallestExtension:DerivativeChanges}]
		To show \eqref{eq:LowerBoundOnRightDerivativeOfEllV1}, let $I'$ be the set of indices of edges active immediately after $\theta$, i.e. 
			\[I' \coloneqq \Set{i \in  I(\theta) | \rDeriv{h_i}(\theta) = \rDeriv{\ell_v}(\theta)}.\]
		Since the total outflow from node $v$ is constant during $[\theta_1,\hat{\theta})$ and flow may only enter edges $vw_i$ with $i \in I'$ after $\theta$, there exists some $j \in I'$, where the inflow rate into $vw_j$ after $\theta$ is the same or larger than before. But then we have $\rDeriv{h_j}(\theta) \geq \lDeriv{h_j}(\theta)$ and, thus,
			\[\min\set{\lDeriv{h_i}(\theta) | i \in I(\theta)} \leq \min\set{\lDeriv{h_i}(\theta) | i \in I'} \leq \lDeriv{h_j}(\theta) \leq \rDeriv{h_j}(\theta) = \rDeriv{\ell_v}(\theta).\]
		If, additionally, no edge becomes newly active at time $\theta$, all edges $vw_i$ with $i \in I'$ have been active directly before $\theta$ as well implying 
			\[\lDeriv{\ell_v}(\theta) = \min\set{\lDeriv{h_i}(\theta) | i \in I(\theta)} \overset{\eqref{eq:LowerBoundOnRightDerivativeOfEllV1}}{\leq} \rDeriv{\ell_v}(\theta). \qedhere\]	
	\end{proofClaim}

	We also need the following observation which is an immediate consequence of the way the water filling algorithm determines the flow distribution (see \Cref{appobs:FlowDistributionDependence}) combined with the lemma's assumption that all label functions $\ell_{w_i}$ have constant derivative during the interval $[\theta_1,\theta_2)$.

	\begin{claim}\label{claim:SmallestExtension:OnlyFinitelyManySubphases}
		There are uniquely defined numbers $\ell_{I,J}$ for all subsets $J \subseteq I \subseteq [p]$ such that $\ell_v'(\theta) = \ell_{I,J}$ within all local phases, where $\set{vw_i | i \in I}$ is the set of active edges in $\edgesLeaving{v}$ and $\set{vw_i | i \in J}$ is the subset of such active edges that also have a non-zero queue during this local phase.\renewcommand\qedsymbol{\ensuremath{\blacksquare}}\qed
	\end{claim}
	
	Using these properties we can now first show a claim which implies that the smallest $\ell_{I,J}$ can only be the derivative of $\ell_v$ for a finite number of intervals. Inductively the same then holds for all of the finitely many $\ell_{I,J}$. The proof of the lemma finally concludes by observing that an interval with constant derivative of $\ell_v$ can contain only finitely many local phases.
	
	\begin{claim}\label{claim:SmallestExtension:LowestDerivativeIntervalsStartWithDifferentEdges}
		Let $(a_1,b_1),(a_2,b_2) \subseteq [\theta_1, \hat{\theta})$ be two disjoint maximal non-empty intervals with constant $\ell'_v(\theta) \eqqcolon c$. If $b_1 < a_2$ and $\ell'_v(\theta) \geq c$ for all $\theta \in (b_1,a_2)$ where the derivative exists, then there exists an edge $vw_i$ such that
		\begin{enumerate}
			\item the first local phase of $(a_2,b_2)$ begins because $vw_i$ becomes newly active and
			\item this edge is not active for any time in the interval $[a_1,a_2)$. 
		\end{enumerate}
		In particular, the first local phase of $(a_1,b_2)$ is not triggered by $vw_i$ becoming active.
	\end{claim}

	\begin{proofClaim}[Proof of \Cref*{claim:SmallestExtension:LowestDerivativeIntervalsStartWithDifferentEdges}]
		Since we have $\rDeriv{\ell_v}(a_2) = c$, \Cref{claim:SmallestExtension:DerivativeChanges} implies that there exists some edge $vw_i$ with $h_i(a_2) = \ell_v(a_2)$ and $\lDeriv{h_i}(a_2) \leq c$. As $(a_2,b_2)$ was chosen to be maximal and $\ell'_v(\theta) \geq c$ holds almost everywhere between $b_1$ and $a_2$, we have $\lDeriv{\ell_v}(a_2) > c$. Thus, $vw_i$ was inactive before $a_2$.
		
		Now let $\tilde{\theta} < a_2$ be the last time before $a_2$, where $vw_i$ was active. By \Cref{claim:SmallestExtension:InactiveHareConvex} we know then that $h'_i(\theta) \leq c$ holds almost everywhere on $[\tilde{\theta},a_2]$. At the same time we have $\ell'_v(\theta) \geq c$ almost everywhere on $[a_1,a_2]$ and $\ell'_v(\theta) > c$ for at least some proper subinterval of $[b_1,a_2]$, since the intervals $(a_1,b_1)$ and $(a_2,b_2)$ were chosen to be maximal. Combining these two facts with $\ell_v(a_2) = h_i(a_2)$ implies $\ell_v(\theta) < h_i(\theta)$ for all $\theta \in [\tilde{\theta},a_2) \cap [a_1,a_2)$. As both functions are continuous we must have $\tilde{\theta} < a_1$. Thus, $vw_i$ is inactive for all of $[a_1,a_2)$.
	\end{proofClaim}

	This claim directly implies that the lowest derivative of $\ell_v$ during $[\theta_1,\hat{\theta}]$ only appears in a finite number of intervals, as each of these intervals has to start with a different edge becoming newly active. But then, iteratively applying this claim for the intervals between these intervals shows that any derivative of $\ell_v$ can only appear in a finite number of intervals. Since, by \Cref{claim:SmallestExtension:OnlyFinitelyManySubphases}, $\ell'_v$ can only attain a finite number of values, this implies that $[\theta_1,\hat{\theta})$ consists of a finite number of intervals with constant derivative of $\ell_v$.
	
	\begin{claim}\label{claim:SmallestExtension:ConstantDerivativeIntervalsHaveFinitelyManySubPhases}
		Let $(a,b) \subseteq [\theta_1,\hat{\theta})$ be an interval during which $\ell'_v$ is constant. Then $(a,b)$ contains at most $2p$ local phases, where $p$ denotes the out-degree of $v$.
	\end{claim}

	\begin{proofClaim}[Proof of \Cref*{claim:SmallestExtension:ConstantDerivativeIntervalsHaveFinitelyManySubPhases}]
		By \Cref{claim:SmallestExtension:InactiveHareConvex} an edge that changes from active to inactive during the interval $(a,b)$ will remain inactive for the rest of this interval. Thus, at most $p$ local phases can start because an edge becomes newly active. By \Cref{claim:SmallestExtension:LowestDerivativeIntervalsStartWithDifferentEdges} if a local phase begins because the queue on an active edge $vw_i$ runs empty at time $\theta$, we have 
		$\rDeriv{h_i}(\theta) > \lDeriv{h_i}(\theta) = \lDeriv{\ell_v}(\theta) = \rDeriv{\ell_v}(\theta)$
		meaning that this edge will become inactive. Thus, at most $p$ local phases start because the queue of an active edge runs empty. Since by observation \ref{obs:TriggerNewPhase} these are the only ways to start a new local phase, we conclude that there can be no more than $2p$ local phases during $(a,b)$.
	\end{proofClaim}
	
	Combining \Cref{claim:SmallestExtension:LowestDerivativeIntervalsStartWithDifferentEdges,claim:SmallestExtension:ConstantDerivativeIntervalsHaveFinitelyManySubPhases} we see that $[\theta_1,\hat{\theta})$ only contains a finite number of local phases and, thus, we achieve $\hat{\theta} = \theta_2$ with finitely many extensions. 
\end{proof}

With this lemma the proof of the following theorem is straightforward.

\begin{theorem}\label{thm:TerminationOfExtensionAlg}
	For any acyclic single-sink network with piecewise constant network-inflow rates an IDE can be constructed in finite time using \Cref{alg:Construction}. 
\end{theorem}

\begin{proof}
	First, note that by \cite[Theorem 1]{GH20PoA} for any given single-sink network $\network$ there exists an (easily computable) time $T$ such that all IDE in $\network$ terminate before $T$. This makes the first line of \Cref{alg:Construction} possible. Thus, it remains to show that in \cref{alg:Construction:FlowDistribution} a finite number of local phases always suffices. We show this by induction over $\theta$ and $i \in [n]$, i.e. we can assume that the currently constructed flow $f$ is a partial IDE up to time $\theta$ for all nodes $v_j, j \geq i$ and up to time $\theta+\tau_{\min}$ for all nodes $v_j, j < i$ with only a finite number of (local) phases. In particular, this means that we can partition the interval $[\theta,\theta+\tau_{\min})$ into a finite number of proper subintervals such that within each such subinterval there is a constant gross node inflow rate into node $v_i$ and the labels at all the vertices $w$ with $v_iw \in \edgesLeaving{v_i}$ change linearly. Then, by \Cref{lemma:SmallestExtension}, we can distribute the flow at node $v_i$ to the outgoing edges using a finite number of local phases for each of these subintervals. Note that, aside from the queue lengths on the edges leaving $v_i$, the so distributed flow has no influence on the flow distribution in later subintervals and, in particular, does not influence the partition into subintervals or the flow distribution at nodes closer to $t$ than $v_i$. Thus, we can distribute the outflow from $v_i$ for the whole interval $[\theta,\theta+\tau_{\min})$ using only a finite number of local phases.
\end{proof}

Closer inspection of the proofs above also allows us to derive a rough but explicit bound on the number of phases the constructed IDE flow can have.

\begin{prop}\label{prop:AlgorithmAcyclicExplicitBound}
	For any acyclic single-sink network with piecewise constant network-inflow rates the number of phases of any IDE flow constructed by \Cref{alg:Construction} is bounded by
		\[\BigO\left(P\Big(2(\Delta+1)^{4^\Delta+1}\Big)^{\abs{V}T/\tau_{\min}}\right),\]
	 where $\Delta \coloneqq \max\set{\abs{\edgesLeaving{v}} | v \in V}$ the maximum out-degree in the given network and $P$ is the number of intervals with constant network inflow rates.
\end{prop}

\begin{proof}
	First, we look at an interval $[\theta_1,\theta_2)$ and a single node $v$ as in \Cref{lemma:SmallestExtension}. Here we can use \Cref{claim:SmallestExtension:LowestDerivativeIntervalsStartWithDifferentEdges} to  bound the number of intervals of constant derivative of $\ell_v$ by 
		\[\left(\abs{\edgesLeaving{v}}+1\right)^{\abs{\set{(I,J) | J \subseteq I \subseteq [\abs{\edgesLeaving{v}}]}}} \leq \left(\abs{\edgesLeaving{v}}+1\right)^{4^{\abs{\edgesLeaving{v}}}},\]
	each of them containing at most $2\abs{\edgesLeaving{v}}$ local phases by \Cref{claim:SmallestExtension:ConstantDerivativeIntervalsHaveFinitelyManySubPhases}. Together this shows that any such interval will be subdivided into at most  $2(\Delta+1)^{4^{\Delta}+1}$ local phases. Thus, whenever we execute \cref{alg:Construction:FlowDistribution} of \Cref{alg:Construction} every currently existing (local) phase may be subdivided further into at most $2(\abs{\Delta}+1)^{4^\Delta+1}$ local phases. Consequently, for every pass of the outer for-loop the number of local phases can be multiplied by at most $\prod_{v \in V}\Big(2(\abs{\Delta}+1)^{4^\Delta+1}\Big)$ in total during the extension over the interval $[\theta,\theta+\tau_{\min})$. Combining this with the at most $P$ phases triggered by changing network inflow rates results in the bound of 
		\[\BigO\left(P\Big(2(\Delta+1)^{4^\Delta+1}\Big)^{\abs{V}T/\tau_{\min}}\right).\qedhere\]
\end{proof}

\subsection{General Single-Sink Networks}

We now want to extend this result to general single-sink networks, i.e. we want to show that \Cref{alg:ConstructionOld} terminates within finite time not only for acyclic graphs, but for all graphs. We first note that the requirement for input-graphs of \Cref{alg:Construction} to be acyclic is somewhat too strong. It is actually enough to have some (static) order on the nodes such that it is always a topological order with respect to the active subgraph. That is, for a general single-sink network we can still apply \Cref{alg:Construction} to determine an IDE-extension with finitely many phases for any interval during which we have such a static node ordering. Thus, \Cref{alg:ConstructionOld} will also use finitely many extension phases for each interval with such a static ordering. This observation gives rise to \Cref{alg:ConstructionGen}, another slight variant of \Cref{alg:ConstructionOld}.

\begin{algorithm}\caption{IDE-Construction Algorithm for general single-sink networks}\label{alg:ConstructionGen}
	\KwIn{A single-sink network $\network$ with piecewise constant network inflow rates}
	
	\KwOut{An IDE flow $f$ in $\network$}
	
	Choose $T$ large enough such that all IDE flows in $\network$ terminate before $T$
	
	Let $f$ be the zero flow, $\theta \leftarrow 0$ and $\tilde{E} \leftarrow E_0$
	
	Determine a topological order $t=v_1 < v_2 < \dots < v_n$ w.r.t. the edges in $\tilde{E}$
	
	\While{$\theta < T$}{
		\SetNoFillComment
		\tcc{$f$ is a partial IDE up to time $\theta$}
		
		\For{$i = 2, \dots, n$}{
			Compute $b_{v_i}^-(\theta)$ and determine a constant distribution of this inflow to edges in $\edgesLeaving{v_i}$ such that the used edges remain active for some proper interval
		}
		Determine the largest $\alpha \geq 0$ such that all $b_v^-$ are constant on $(\theta,\theta+\alpha)$ and the set of active edges does not change
		
		Extend $f$ up to time $\theta+\alpha$ with constant edge inflow rates and set $\theta \leftarrow \theta+\alpha$
		
		\If{$E_{\theta} \setminus \tilde{E} \neq \emptyset$}{
			Define $\tilde{E} \leftarrow \tilde{E} \cup E_{\theta}$.\label{alg:ConstructionGen:AddNewlyActiveEdges}
			
			\While{there exists a cycle $C$ in $\tilde{E}$}{
				Remove an edge $e=xy$ with the largest value $\ell_y(\theta) - \ell_x(\theta)$ of all edges in $C$\label{alg:ConstructionGen:EdgeRemoval}
			}
			
			Determine a topological order $t=v_1 < v_2 < \dots < v_n$ w.r.t. the edges in $\tilde{E}$\label{alg:ConstructionGen:TopOrder}	
		}
	}
\end{algorithm}

We will prove that this algorithm does indeed construct an IDE for arbitrary single-sink networks within finite time by first showing that this algorithm is a special case of the original algorithm. Thus, it is correct and uses only a finite number of phases for any interval in which the topological order does not change. We then conclude the proof by showing that it is indeed enough to change the topological order a  finite number of times for any given time horizon.

\begin{lemma}\label{lemma:AlgGenIsCorrect}
	\Cref{alg:ConstructionGen} is a special case of \Cref{alg:ConstructionOld}. In particular it is correct.
\end{lemma} 

\begin{proof}
	As in \Cref{alg:Construction} the existence of an upper bound $T$ on the termination time of all IDE flows for a given single-sink network is guaranteed by \cite[Theorem 1]{GH20PoA}. Next, note that $\tilde{E}$ is clearly always acyclic (except in \cref{alg:ConstructionGen:AddNewlyActiveEdges,alg:ConstructionGen:EdgeRemoval}) which guarantees that we can always find a topological order with respect to $\tilde{E}$. We now only need to show that such an ordering is also a topological order with respect to the active edges, i.e. that for any time $\theta$ we have $E_{\theta} \subseteq \tilde{E}$. For this we will use the following observation
	
	\begin{claim}\label{claim:AlgGenNeverRemovesActiveEdges}
		Any edge $xy$ removed from $\tilde{E}$ in \cref{alg:ConstructionGen:EdgeRemoval} of \Cref{alg:ConstructionGen} satisfies $\ell_{x}(\theta) < \ell_{y}(\theta)$.
	\end{claim}

	\begin{proofClaim}
		Let $C \subseteq \tilde{E}$ be a cycle containing the removed edge $xy$. Since $\tilde{E}$ was acyclic before we added the newly active edges in \cref{alg:ConstructionGen:AddNewlyActiveEdges}, this cycle also has to contain some currently active edge $vw$. This gives us
			\begin{align*}
				\sum_{e=uz \in C\setminus\{vw\}}(\ell_z(\theta)-\ell_u(\theta))	&= \sum_{e=uz \in C}(\ell_z(\theta)-\ell_u(\theta)) - (\ell_w(\theta) - \ell_v(\theta))  \\
					&= 0 - \ell_w(\theta) + \Big(\ell_w(\theta) + \tau_{vw} + \frac{q_{vw}(\theta)}{\nu_{vw}}\Big) \\
					&= \tau_{vw} + \frac{q_{vw}(\theta)}{\nu_{vw}} \geq \tau_{\min}. 
			\end{align*}
		Thus, $C$ contains at least one edge $uz$ with $\ell_z(\theta)-\ell_u(\theta) > 0$ and, by the way it was chosen, this then holds in particular for edge $xy$.
	\end{proofClaim}
	
	This claim immediately implies that in \cref{alg:ConstructionGen:EdgeRemoval} we only remove inactive edges and that, afterwards, we still have $E_{\theta} \subseteq \tilde{E}$.
\end{proof}

\begin{lemma}\label{lemma:AlgGenETildeChangesNotToOften}
	For any single-sink network there exists some constant $C > 0$ such that for any time interval of length $C$ the set $\tilde{E}$ changes at most $\abs{E}$ times during this interval in \Cref{alg:ConstructionGen}.
\end{lemma}

\begin{proof}
	The proof of this \namecref{lemma:AlgGenETildeChangesNotToOften} mainly rest on the following claim stating that for any fixed network we can bound the slope of the node labels of any feasible flow in this network by some constant.
	
	\begin{claim}\label{claim:BoundOnLabelSlopes}
		For any given network there exists some constant $L > 0$ such that for all feasible flows, all nodes $v$ and all times $\theta$ we have $\abs{\ell_v'(\theta)} \leq L$.
	\end{claim}

	\begin{proofClaim}
		First note that for any node $v$ we can bound the maximal inflow rate into this node by some constant $L_v$ as follows:	
			\[\sum_{e \in \edgesEntering{v}}f^-_e(\theta) + u_v(\theta) \overset{\text{\eqref{eq:FeasibleFlow-QueueOpAtCap}}}{\leq} \sum_{e \in \edgesEntering{v}}\nu_e + \max\set{u_v(\theta) | \theta \in \IR_{\geq 0}} \eqqcolon L_v.\]
		Using flow conservation \eqref{eq:FeasibleFlow-FlowConservation} this, in turn, allows us to bound the inflow rates into all edges $e \in \edgesLeaving{v}$ and, thus, the rate at which the queue length and the current travel time on these edges can change:
			\[-1 \leq c'_e(\theta) \overset{\text{\eqref{Def:QueueLength},\eqref{Def:InstantaneousTravelTime}}}{\leq} \frac{f^+_e(\theta)}{\nu_e} \leq \frac{L_v}{\nu_e} \eqqcolon L_e.\]
		Since this rate of change is also lower bounded by $-1$  setting $L \coloneqq \sum_{e \in E}\max\set{1,L_e}$ proves the claim, as for all nodes $v$ and times $\theta$ we then have
			\[\abs{\ell'_v(\theta)} \leq \sum_{e \in E}\abs{c'_e(\theta)} \leq \sum_{e \in E}L_e = L.\qedhere\]		
	\end{proofClaim}

	Now, from \Cref{claim:AlgGenNeverRemovesActiveEdges} we know that, whenever we remove an edge $xy$ from $\tilde{E}$ at time $\theta$ we must have $\ell_{x}(\theta) < \ell_{y}(\theta)$. But at the time where we last added this edge to $\tilde{E}$, say at time $\theta' < \theta$, it must have been active (since we only ever add active edges to $\tilde{E}$) and, thus, we had $\ell_x(\theta') = \ell_y(\theta') + c_{xy}(\theta') \geq \ell_y(\theta') + \tau_{\min}$. Therefore, the difference between the labels at $x$ and $y$ has changed by at least by $\tau_{\min}$ between $\theta'$ and $\theta$. \Cref{claim:BoundOnLabelSlopes} then directly implies $\theta-\theta' \geq \frac{\tau_{\min}}{2L}$. So, for any time interval of length at most $\frac{\tau_{\min}}{2L}$ each edge can be added at most once to $\tilde{E}$. Since $\tilde{E}$ only ever changes when we add at least one new edge to it, setting $C \coloneqq \frac{\tau_{\min}}{2L}$ proves the \namecref{lemma:AlgGenETildeChangesNotToOften}.
\end{proof}

\begin{theorem}\label{thm:TerminationOfGenConstructionAlg}
	For any single-sink network with piecewise constant network-inflow rates an IDE can be constructed in finite time using \Cref{alg:ConstructionGen}.
\end{theorem}

\begin{proof}
	By \Cref{lemma:AlgGenIsCorrect} \Cref{alg:ConstructionGen} is a special case of \Cref{alg:ConstructionOld}. Thus, for any interval with static $\tilde{E}$ it produces the same flow as \Cref{alg:Construction}. In particular, by \Cref{thm:TerminationOfExtensionAlg}, for any such interval the constructed flow consists of finitely many phases. Finally, \Cref{lemma:AlgGenETildeChangesNotToOften} shows that the whole relevant interval $[0,T]$ can be partitioned into a finite number of intervals with static set $\tilde{E}$. Consequently, \Cref{alg:ConstructionGen} constructs an IDE  with finitely many phases and, thus, terminates within finite time.
\end{proof}

As in the acyclic case we can again also extract an explicit upper bound on the number of phases.

\begin{prop}\label{prop:AlgorithmGeneralExplicitBound}
	For any single-sink network with piecewise constant network inflow rates the number of phases of any IDE flow constructed by \Cref{alg:ConstructionGen} is bounded by
		\[\BigO\left(P\Big(2(\Delta+1)^{4^\Delta+1}\Big)^{2L\cdot \abs{E}\cdot \abs{V}\cdot T / \tau_{\min}^2}\right),\]
	where, again, $\Delta \coloneqq \max\set{\abs{\edgesLeaving{v}} | v \in V}$ is the maximum out-degree in the given network, $P$ is the number of intervals with constant network inflow rates and $L$ the bound on the slopes of the label functions from \Cref{claim:BoundOnLabelSlopes}.
\end{prop}

\begin{proof}
	For any time interval with fixed node order \Cref{alg:ConstructionGen} is equivalent to \Cref{alg:Construction} and, thus, the bound from \Cref{prop:AlgorithmAcyclicExplicitBound} applies. Also note, that in \Cref{alg:Construction} we could change the node order after every time step of length $\tau_{\min}$ without any impact on correctness or the bound on the number of phases (as long as we always choose an order which is a topological order with respect to the active edges). As, by \Cref{lemma:AlgGenETildeChangesNotToOften}, the node order in \Cref{alg:ConstructionGen} changes at most $2L\cdot\abs{E}/\tau_{\min}$ times during any unit time interval, replacing $T$ by $2L\cdot\abs{E}\cdot T/\tau_{\min}$ in the bound for \Cref{alg:Construction} yields a valid bound for the number of phases of \Cref{alg:ConstructionGen}.
\end{proof}

\begin{remark}
	If presented with rational input data (i.e. rational capacities, node inflow rate, current queue lengths, current distance labels and slopes of distance labels of neighbouring nodes) the water filling procedure \Cref{Alg:FindExtension} again produces a rational solution to \eqref{OPT:bvThetaK} (i.e. rational edge inflow rates) which then, in turn, results in a rational maximal extension length $\alpha$. Thus, \Cref{alg:ConstructionGen} can be implemented as an exact combinatorial algorithm.
\end{remark}

Since DE and IDE coincide for parallel link networks and for DE paths can always be replaced by single edges, the above theorem also implies the following result for DE. Note, however, that, while to the best of our knowledge this result has never explicitly been stated elsewhere, it seems very likely that it could also be shown in a more direct way for this very simple graph class.

\begin{cor}
	On parallel paths networks Dynamic Equilibria can be constructed in finite time using the natural extension algorithm.
\end{cor}


\section{Computational Complexity of IDE-Flows}

While \Cref{thm:TerminationOfGenConstructionAlg} shows that IDE flows can be constructed in finite time, the bound provided in \Cref{prop:AlgorithmGeneralExplicitBound} is clearly superpolynomial. We now want to show that in some sense this is to be expected. Namely, we first look at the output complexity of any such algorithm, i.e. how complex the structure of IDE flows can be. Then we show that many natural decision problems involving IDE are actually NP-hard.

\subsection{Output Complexity and Steady State}\label{sec:SteadyStateCC}

In this section we call an open interval $(a,b) \subseteq \IR_{\geq 0}$ a \emph{phase} of a feasible flow $f$, if it is a maximal interval with constant in- and outflow rates for all edges. Then it seems reasonable to expect of any algorithm computing feasible flows that its output has to contain in some way a list of the flow's phases and corresponding in- and outflow rates. In particular, the number of phases of a flow is a lower bound for the runtime of any algorithm determining that flow. This observation allows us to give an exponential lower bound for the output complexity and therefore also for the worst case runtime of any algorithm determining IDE flows. This remains true even if we only look at acyclic graphs and allow for our algorithm to recognize simple periodic behaviour and abbreviate the output accordingly.

\begin{theorem}\label{thm:Complexity}
	The worst case output complexity of calculating IDE flows is not polynomial in the encoding size of the instance, even if we are allowed to use periodicity to reduce the encoding size of the determined flow. This is true even for series parallel graphs.
\end{theorem}

\begin{proof}
	For any given $U \in \INs$ consider the network pictured in \Cref{fig:CompComplexity} with a constant inflow rate of $2$ at $s$ over the interval $[0,U]$. This network can clearly be encoded in $\bigO(\log U)$ space. The unique (up to changes on a set of measure zero) IDE is displayed up to time $\theta = 6.5$ in \Cref{fig:CompComplexity} and described for all times in \Cref{tab:CompComplexity}. As this pattern is clearly non-periodic and continues up to time $\theta=U$, it exhibits $\bigOm(U)$ distinct phases. This proves the theorem.
	\begin{figure}[ht]\centering
		\begin{adjustbox}{max width=\textwidth}
			\begin{tikzpicture}[scale=1.2]
			
	\newcommand{\exampleGraph}[2]{		
		\node[labeledNodeS] (s) at (0, 4) {$s$};
		\draw (s) ++(4, 0) node[labeledNodeS] (v) {$v$};
		\draw (s) ++(0, -4) node[labeledNodeS] (w) {$w$};
		\draw (w) ++(2, 0) node[labeledNodeS, fill=white] (x) {$x$};
		\draw (v) ++(0, -4) node[labeledNodeS] (t) {$t$};
		
		\draw[normalEdge,->] (s) -- node[below] {\ifthenelse{\equal{#2}{labels}}{\Large$(\tau_{sv},\nu_{sv})=(1,2)$}{}} (v);
		\draw[normalEdge,->] (s) -- node[right] {\ifthenelse{\equal{#2}{labels}}{\Large$(1,2)$}{}} (w);
		\draw[normalEdge,->] (v)  -- node[left] {\ifthenelse{\equal{#2}{labels}}{\Large$(1,1)$}{}} (t);
		\draw[normalEdge,->] (w) -- node[below] {\ifthenelse{\equal{#2}{labels}}{\Large$(1,1)$}{}} (x);
		\draw[normalEdge,->] (x) -- node[below] {\ifthenelse{\equal{#2}{labels}}{\Large$(1,1)$}{}} (t);
		
		\node at (2,5) [rectangle,draw] (theta0) {\Large$\theta=#1$:};

		\draw (s) +(0,1) node[rectangle, draw, minimum width=1.5cm, minimum height=1cm,fill=\colComRed] (su) {\Large $u_1\equiv 2$}; 
		\draw [->, line width=2pt] (su) -- (s);
	}
	
	\newcommand{\colComRed}{red!30}
	\newcommand{\colComBlue}{blue!60}
	
	\begin{scope}
	
		\exampleGraph{0}{labels}
		
		\end{scope}

	\begin{scope}[xshift=7cm]
		
		\node[rectangle, color=\colComRed,draw, minimum width=4cm,minimum height=.4cm, fill=\colComRed] at (2,4) {};
		
		\exampleGraph{1}{}		
	\end{scope}	

	\begin{scope}[xshift=14cm]
		
		\node[rectangle, color=\colComRed,draw, minimum width=4cm,minimum height=.4cm, fill=\colComRed] at (2,4) {};
		\node[rectangle, color=\colComRed,draw, minimum width=4cm,minimum height=.2cm, fill=\colComRed,rotate around={90:(0,0)}] at (4,2) {};
		\node[rectangle, color=\colComRed, draw, minimum width=.8cm,minimum height=1cm, fill=\colComRed,rotate around={90:(0,0)}] at (4.5,3.2) {};
		\node at (5,2.5) [] () {\Large $q_{vt}(2)=1$};
		
		\exampleGraph{2}{}		
	\end{scope}	
	
	\begin{scope}[xshift=21cm]

		\node[rectangle, color=\colComRed,draw, minimum width=4cm,minimum height=.2cm, fill=\colComRed,rotate around={90:(0,0)}] at (4,2) {};
		\node[rectangle, color=\colComRed, draw, minimum width=.8cm,minimum height=2cm, fill=\colComRed,rotate around={90:(0,0)}] at (4.95,3.2) {};
		\node at (5,2.5) [] () {\Large $q_{vt}(3)=2$};
		
		\node[rectangle, color=\colComRed,draw, minimum width=4cm,minimum height=.4cm, fill=\colComRed,rotate around={90:(0,0)}] at (0,2) {};
		
		\exampleGraph{3}{}		
	\end{scope}		
	
	\begin{scope}[yshift=-6.5cm]
		
		\node[rectangle, color=\colComRed,draw, minimum width=4cm,minimum height=.2cm, fill=\colComRed,rotate around={90:(0,0)}] at (4,2) {};
		\node[rectangle, color=\colComRed, draw, minimum width=.8cm,minimum height=1.5cm, fill=\colComRed,rotate around={90:(0,0)}] at (4.75,3.2) {};
		\node at (5.4,2.5) [] () {\Large $q_{vt}(3.5)=1.5$};
		
		\node[rectangle, color=\colComRed,draw, minimum width=4cm,minimum height=.4cm, fill=\colComRed,rotate around={90:(0,0)}] at (0,2) {};
		\node[rectangle, color=\colComRed,draw, minimum width=1cm,minimum height=.2cm, fill=\colComRed] at (.8,0) {};
		\node[rectangle, color=\colComRed, draw, minimum width=.8cm,minimum height=.5cm, fill=\colComRed,rotate around={180:(0,0)}] at (0.8,-0.3) {};
		\node at (2.5,-.6) [] () {\Large $q_{wx}(3.5)=0.5$};		
				
		\exampleGraph{3.5}{}		
	\end{scope}	
	
	\begin{scope}[yshift=-6.5cm,xshift=7cm]

		\node[rectangle, color=\colComRed,draw, minimum width=4cm,minimum height=.4cm, fill=\colComRed] at (2,4) {};
		\node[rectangle, color=\colComRed,draw, minimum width=4cm,minimum height=.2cm, fill=\colComRed,rotate around={90:(0,0)}] at (4,2) {};
		\node[rectangle, color=\colComRed, draw, minimum width=.8cm,minimum height=.5cm, fill=\colComRed,rotate around={90:(0,0)}] at (4.3,3.2) {};
		\node at (5.4,2.5) [] () {\Large $q_{vt}(4.5)=0.5$};
		
		\node[rectangle, color=\colComRed,draw, minimum width=3cm,minimum height=.2cm, fill=\colComRed] at (1.6,0) {};
		\node[rectangle, color=\colComRed, draw, minimum width=.8cm,minimum height=1.5cm, fill=\colComRed,rotate around={180:(0,0)}] at (0.8,-.7) {};
		\node at (2.5,-.6) [] () {\Large $q_{wx}(4.5)=1.5$};	
		
		\exampleGraph{4.5}{}		
	\end{scope}		
	
	\begin{scope}[yshift=-6.5cm,xshift=14cm]
		
		\node[rectangle, color=\colComRed,draw, minimum width=4cm,minimum height=.4cm, fill=\colComRed] at (2,4) {};
		\node[rectangle, color=\colComRed,draw, minimum width=4cm,minimum height=.2cm, fill=\colComRed,rotate around={90:(0,0)}] at (4,2) {};
		\node[rectangle, color=\colComRed, draw, minimum width=.8cm,minimum height=1.5cm, fill=\colComRed,rotate around={90:(0,0)}] at (4.75,3.2) {};
		\node at (5.4,2.5) [] () {\Large $q_{vt}(5.5)=1.5$};
		
		\node[rectangle, color=\colComRed,draw, minimum width=4cm,minimum height=.2cm, fill=\colComRed] at (2,0) {};
		\node[rectangle, color=\colComRed, draw, minimum width=.8cm,minimum height=.5cm, fill=\colComRed,rotate around={180:(0,0)}] at (0.8,-0.3) {};
		\node at (2.5,-.6) [] () {\Large $q_{wx}(5.5)=0.5$};	
				
		\exampleGraph{5.5}{}		
	\end{scope}		

	\begin{scope}[yshift=-6.5cm,xshift=21cm]
		
		\node[rectangle, color=\colComRed,draw, minimum width=4cm,minimum height=.2cm, fill=\colComRed,rotate around={90:(0,0)}] at (4,2) {};
		\node[rectangle, color=\colComRed, draw, minimum width=.8cm,minimum height=2.5cm, fill=\colComRed,rotate around={90:(0,0)}] at (5.1,3.2) {};
		\node at (5.4,2.5) [] () {\Large $q_{vt}(6.5)=2.5$};

		\node[rectangle, color=\colComRed,draw, minimum width=4cm,minimum height=.4cm, fill=\colComRed,rotate around={90:(0,0)}] at (0,2) {};		
		\node[rectangle, color=\colComRed,draw, minimum width=3cm,minimum height=.2cm, fill=\colComRed] at (2.4,0) {};
		
		\exampleGraph{6.5}{}		
	\end{scope}	
\end{tikzpicture}
		\end{adjustbox}
		\caption{A network (top left picture) where constant inflow rate of $2$ over $[0,U]$ leads to an IDE with $\bigOm(U)$ different phases. The following pictures show the first states of the network, which are described in general in \Cref{tab:CompComplexity}.}\label{fig:CompComplexity}
	\end{figure}
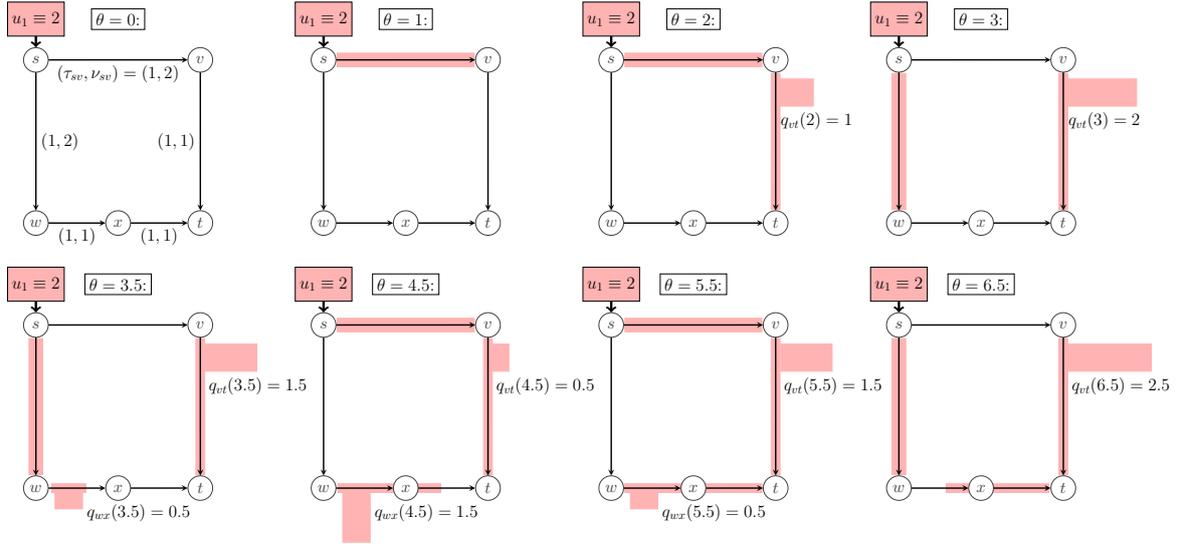
	\begin{table}[h]\centering
		\begin{tabular}{l|c|c|cc|cc|c|c}
			$\theta=$ 		& $f^+_{vt}(\theta)$ & $f^+_{wx}(\theta)$ & $q_{vt}(\theta)$ 	&& $q_{wx}(\theta)$ 		&& $f^+_{sv}(\theta)$ 	& $f^+_{sw}(\theta)$ \\\hline
			$4k+2^{-k}-1$ 	& $0$				 & $2$				  & $2-2^{-k}$ & $\searrow$	& $1-2^{-k}$ & $\nearrow$ 	& $2$					& $0$ \\
			$4k+2^{-k}$ 	& $2$				 & $0$				  & $1-2^{-k}$ & $\nearrow$	& $2-2^{-k}$ & $\searrow$ 	& $2$					& $0$ \\
			$4k+2^{-k}+1$ 	& $2$				 & $0$				  & $2-2^{-k}$ & $\nearrow$	& $1-2^{-k}$ & $\searrow$ 	& $0$					& $2$ \\
			$4k+2$ 			& $2$				 & $0$				  & $2-2^{-k}$ & $\nearrow$	& $0$ & $\rightarrow$ 	& $0$					& $2$ \\
			$4k+2^{-k}+2$ 	& $0$				 & $2$				  & $3-2^{-k}$ & $\searrow$	& $0$ & $\nearrow$ 			& $0$					& $2$ \\
		\end{tabular}
		\caption{Phases of the (unique) IDE in the instance of \Cref{fig:CompComplexity}. For all $k \in \INo$ the table includes the (constant) inflow rates into edges on the intervals $(4k+2^{-k}-1,4k+2^{-k})$, $(4k+2^{-k},4k+2^{-k}+1)$,$(4k+2^{-k}+1,4k+2)$, $(4k+2,4k+2^{-k}+2)$ and $(4k+2^{-k}+2,4k+2^{-(k+1)}+3)$ as well as the queue lengths on the edges $vt$ and $wx$ at the beginning of these intervals and the rate of change for the queue lengths over the following interval ($\nearrow$ stands for an increase at rate $1$, $\searrow$ for a decrease at rate $-1$ and $\rightarrow$ for no change).}\label{tab:CompComplexity}
	\end{table}
\end{proof}

\begin{remark}
	In \cite[Section 5.2]{CominettiCO20} \citeauthor*{CominettiCO20} sketch a family of instances of size $\bigO(d^2)$ where a dynamic equilibrium flow exhibits an exponential number of phases (of order $\bigOm(2^d)$) before it reaches a stable state. 
\end{remark}

The network constructed in the above proof can also be used to gain some insights into the long term behavior of IDE flows, i.e. how such flows behave if the inflow rates continue forever. In order to analyze this long term behavior of dynamic equilibrium flows Cominetti~et~al. define in \cite[Section 3]{CominettiCO20} the concept of a steady state:

\begin{defn}
	A feasible flow $f$ with forever lasting constant inflow rate reaches a \emph{steady state} if there exists a time $\tilde{\theta}$ such that after this time all queue lengths stay the same forever i.e.
		\[q_e(\theta) = q_e(\tilde{\theta}) \text{ f.a. } e \in E, \theta \geq \tilde{\theta}.\]
\end{defn}

For dynamic equilibrium flows Cominetti et al. then show that the obvious necessary condition that the inflow rate is at most the minimal total capacity of any $s$-$t$ cut is also a sufficient condition for any dynamic equilibrium in such a network to eventually reach a steady state (\cite[Theorem 3]{CominettiCO20}). We will show that this is not true for IDE flows - even if we consider a weaker variant of steady states:

\begin{defn}
	A feasible flow $f$ reaches a \emph{periodic state} if there exists a time $\tilde{\theta}$ and a periodicity $p \in \IR_{\geq 0}$ such that after time $\tilde{\theta}$ all queue lengths change in a periodic manner, i.e.
	\[q_e(\theta+kp) = q_e(\theta) \text{ f.a. } e \in E, \theta \geq \tilde{\theta}, k \in \INs.\]
\end{defn}

Note that, in particular, every flow reaching a stable state also reaches a periodic state (with arbitrary periodicity).

\begin{theorem}\label{thm:NoSteadyState}
	There exists a series parallel network with a forever lasting constant inflow rate $u$ at a single node $s$, satisfying $u \leq \sum_{e \in \edgesLeaving{X}}\nu_e$ for all $s$-$t$ cuts $X$, where no IDE ever reaches a periodic state.
\end{theorem}

\begin{proof}
	Consider the network constructed in the proof of \Cref{thm:Complexity}, i.e. the one pictured in \Cref{fig:CompComplexity}, but with a constant inflow rate of $2$ at $s$ for all of $\R_{\geq 0}$. A minimal cut is $X = \set{s,v,w}$ with $\sum_{e \in \edgesLeaving{X}}\nu_e = 2$. The unique IDE flow is still the one described in \Cref{tab:CompComplexity} and, thus, never reaches a periodic state.
\end{proof}

\begin{remark}
	In contrast the (again unique) dynamic equilibrium for the network from \Cref{fig:CompComplexity} is displayed in \Cref{fig:SteadyStateNash} and does indeed reach a steady state at time $\theta=4$.
	
	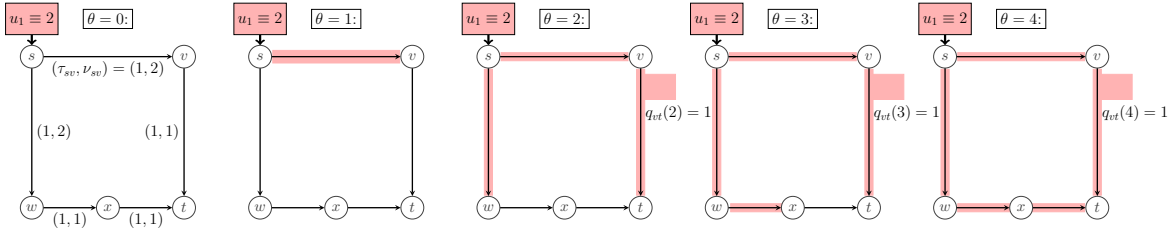
\begin{figure}[ht]\centering
		\begin{adjustbox}{max width=\textwidth}
			\begin{tikzpicture}[scale=1.2]
			
	\newcommand{\exampleGraph}[2]{		
		\node[labeledNodeS] (s) at (0, 4) {$s$};
		\draw (s) ++(4, 0) node[labeledNodeS] (v) {$v$};
		\draw (s) ++(0, -4) node[labeledNodeS] (w) {$w$};
		\draw (w) ++(2, 0) node[labeledNodeS, fill=white] (x) {$x$};
		\draw (v) ++(0, -4) node[labeledNodeS] (t) {$t$};
		
		\draw[normalEdge,->] (s) -- node[below] {\ifthenelse{\equal{#2}{labels}}{\Large$(\tau_{sv},\nu_{sv})=(1,2)$}{}} (v);
		\draw[normalEdge,->] (s) -- node[right] {\ifthenelse{\equal{#2}{labels}}{\Large$(1,2)$}{}} (w);
		\draw[normalEdge,->] (v)  -- node[left] {\ifthenelse{\equal{#2}{labels}}{\Large$(1,1)$}{}} (t);
		\draw[normalEdge,->] (w) -- node[below] {\ifthenelse{\equal{#2}{labels}}{\Large$(1,1)$}{}} (x);
		\draw[normalEdge,->] (x) -- node[below] {\ifthenelse{\equal{#2}{labels}}{\Large$(1,1)$}{}} (t);
		
		\node at (2,5) [rectangle,draw] (theta0) {\Large$\theta=#1$:};

		\draw (s) +(0,1) node[rectangle, draw, minimum width=1.5cm, minimum height=1cm,fill=\colComRed] (su) {\Large $u_1\equiv 2$}; 
		\draw [->, line width=2pt] (su) -- (s);
	}
	
	\newcommand{\colComRed}{red!30}
	\newcommand{\colComBlue}{blue!60}
	
	\begin{scope}
	
		\exampleGraph{0}{labels}
		
		\end{scope}

	\begin{scope}[xshift=6cm]
		
		\node[rectangle, color=\colComRed,draw, minimum width=4cm,minimum height=.4cm, fill=\colComRed] at (2,4) {};
		
		\exampleGraph{1}{}		
	\end{scope}	

	\begin{scope}[xshift=12cm]
		
		\node[rectangle, color=\colComRed,draw, minimum width=4cm,minimum height=.2cm, fill=\colComRed] at (2,4) {};
		\node[rectangle, color=\colComRed,draw, minimum width=4cm,minimum height=.2cm, fill=\colComRed,rotate around={90:(0,0)}] at (4,2) {};
		\node[rectangle, color=\colComRed, draw, minimum width=.8cm,minimum height=1cm, fill=\colComRed,rotate around={90:(0,0)}] at (4.5,3.2) {};
		\node at (5,2.5) [] () {\Large $q_{vt}(2)=1$};
		
		\node[rectangle, color=\colComRed,draw, minimum width=4cm,minimum height=.2cm, fill=\colComRed,rotate around={90:(0,0)}] at (0,2) {};		
		
		\exampleGraph{2}{}		
	\end{scope}	
	
	\begin{scope}[xshift=18cm]
				
		\node[rectangle, color=\colComRed,draw, minimum width=4cm,minimum height=.2cm, fill=\colComRed] at (2,4) {};
		\node[rectangle, color=\colComRed,draw, minimum width=4cm,minimum height=.2cm, fill=\colComRed,rotate around={90:(0,0)}] at (4,2) {};
		\node[rectangle, color=\colComRed, draw, minimum width=.8cm,minimum height=1cm, fill=\colComRed,rotate around={90:(0,0)}] at (4.5,3.2) {};
		\node at (5,2.5) [] () {\Large $q_{vt}(3)=1$};
		
		\node[rectangle, color=\colComRed,draw, minimum width=4cm,minimum height=.2cm, fill=\colComRed,rotate around={90:(0,0)}] at (0,2) {};		
		\node[rectangle, color=\colComRed,draw, minimum width=2cm,minimum height=.2cm, fill=\colComRed] at (1.2,0) {};
				
		\exampleGraph{3}{}			
	\end{scope}		
	
	\begin{scope}[xshift=24cm]
		
		\node[rectangle, color=\colComRed,draw, minimum width=4cm,minimum height=.2cm, fill=\colComRed] at (2,4) {};
		\node[rectangle, color=\colComRed,draw, minimum width=4cm,minimum height=.2cm, fill=\colComRed,rotate around={90:(0,0)}] at (4,2) {};
		\node[rectangle, color=\colComRed, draw, minimum width=.8cm,minimum height=1cm, fill=\colComRed,rotate around={90:(0,0)}] at (4.5,3.2) {};
		\node at (5,2.5) [] () {\Large $q_{vt}(4)=1$};
		
		\node[rectangle, color=\colComRed,draw, minimum width=4cm,minimum height=.2cm, fill=\colComRed,rotate around={90:(0,0)}] at (0,2) {};		
		\node[rectangle, color=\colComRed,draw, minimum width=4cm,minimum height=.2cm, fill=\colComRed] at (2,0) {};
		
		\exampleGraph{4}{}		
	\end{scope}	
\end{tikzpicture}
		\end{adjustbox}
		\caption{The dynamic equilibrium flow for the network constructed in the proof of \Cref{thm:NoSteadyState}.}\label{fig:SteadyStateNash}
	\end{figure}	
\end{remark}

\begin{remark}\label{rem:SmallExtensionPhases}
	The network considered in the proof of \Cref{thm:NoSteadyState} also shows that we can in fact achieve arbitrarily short extension phases even within quite simple networks. Namely, the gross node inflow rate at node $x$ is of the following form
		\[b_x^-(\theta) = \begin{cases}0, &\text{ if } \theta \in [4k+3,4k+3+2^{-k}] \text{ for some } k \in \INo \\ 1, &\text{ else.} \end{cases}\]
	Thus, the flow distribution at node $x$ requires phases of lengths $2^{-k}$ for any $k \in \INo$. Note however, that these ever smaller getting phases are far enough apart so as to still allow us to reach any finite time horizon within a finite number of extension phases (as it is guaranteed by \Cref{thm:TerminationOfExtensionAlg}).
\end{remark}


\subsection{NP-Hardness}\label{sec:Hardness}

We will now show that the decision problem whether in a given network there exists an IDE with certain properties is often NP-hard -- even if we restrict ourselves to only single-source single-sink networks on acyclic graphs. Note, however, that due to the non-uniqueness of IDE flows this does not automatically imply that computing \emph{any} IDE must be hard. 

We first observe that the restriction to a single source can be made without loss of generality. 

\begin{lemma}\label{lemma:wlogSingleSink}
	For any multi-source single-sink network $\network$ with piecewise constant inflow rates with finitely many jump points there exists a (larger) single-source single-sink network $\network'$ with constant inflow rate such that
	\begin{enumerate}[label=\alph*)]
		\item the encoding size of $\network'$ is linearly bounded in that of $\network$,
		\item if $\network$ is acyclic, so is $\network'$,
		\item $\network$ is a subnetwork of $\network'$ (except for the sources),
		\item the restriction map composed with some constant translation is a one-to-one correspondence between the IDE-flows in $\network'$ and those in $\network$:
			\[\set{\text{ IDE in } \network'} \to \set{\text{ IDE in } \network}, f \mapsto f|_{\network}(\_ - c). \]
	\end{enumerate}
\end{lemma}

\begin{proof}
	This can be accomplished by using the construction from the proof of \cite[Theorem 6.3]{GHS18}, which clearly satisfies all four properties.
\end{proof}

\begin{theorem}\label{thm:NPHardness}
	The following decision problems are NP-hard: 
	\begin{enumerate}[label=(\roman*)]
		\item Given a network and a specific edge: Is there an IDE not using this edge?\label{thm:NPHardness:UnsusedEdge}
		\item Given a network and a specific edge: Is there an IDE using this edge?\label{thm:NPHardness:UsedEdge}
		\item Given a network and a time horizon $T$: Is there an IDE that terminates before $T$?\label{thm:NPHardness:EarlyTermination}
		\item Given a network and some $k \in \IN$: Is there an IDE consisting of at most $k$ phases?\label{thm:NPHardness:FewPhases}
	\end{enumerate}
	All these decision problems remain NP-hard even if we restrict them to single-source instances with constant inflow rate on acyclic graphs. Problem \ref{thm:NPHardness:FewPhases} becomes NP-complete if we restrict $k$ by some polynomial in the encoding size of the whole instance.
\end{theorem}

\begin{proof}
	We will show this theorem by reducing the NP-complete problem \ThreeSAT{} to the above problems. The main idea of the reduction is as follows: For any given instance of \ThreeSAT{} we construct a network which contains a source node for each clause with three outgoing edges corresponding to the three literals of the clause. Any satisfying interpretation of the \ThreeSAT-formula translates to a distribution of the network inflow to the literal edges, which leads to an IDE flow that passes through the whole network in a straightforward manner. If, on the other hand, the formula is  unsatisfiable every IDE flow will cause a specific type of congestion which will divert a certain amount of flow into a different part of the graph. This part of the graph may contain an otherwise unused edge or a gadget which produces many phases (e.g. the graph constructed for the proof of \Cref{thm:Complexity}) or a long travel time (e.g. an edge with very small capacity).

	\begin{figure}\centering
		\begin{minipage}[c]{0.35\textwidth}
			\begin{adjustbox}{max width=\textwidth}
				\begin{tikzpicture}
	\newcommand{\colComRed}{red!30}
	
	\draw (0,1) -- (3.4,-2.7) -- (-3.4,-2.7) -- cycle;
	\node () at (1.5,-1) {\Large $C$};

	\node[namedVertex] (s) at (0,0) {$c$};
	\draw (s) +(0,.7) node[rectangle, draw, minimum width=1.5cm, minimum height=1cm,fill=\colComRed,anchor=south] (su) {$u_{s}\equiv 12\cdot\CharF[{[0,1]}]$}; 
	\draw [->] (su) -- (s);
	
	\node[namedVertex] (x) at (-2,-2) {$\ell_1$};
	\node[namedVertex] (y) at (0,-2) {$\ell_2$};
	\node[namedVertex] (z) at (2,-2) {$\ell_3$};
		
	\node[namedVertex,dashed] (t) at (0,-6) {$t$};
	
	\draw[edge,ultra thick] (s) -- (x);
	\draw[edge,ultra thick] (s) -- (y);
	\draw[edge,ultra thick] (s) -- (z);
	\draw[edge,dashed] (x) -- (t);
	\draw[edge,dashed] (y) -- (t);
	\draw[edge,dashed] (z) -- (t);
\end{tikzpicture}
			\end{adjustbox}
		\end{minipage}\hfill
		\begin{minipage}[c]{0.63\textwidth}
			\caption{The clause gadget $C$ consists of a source node and three edges leaving it, each with capacity $12$ and travel time $1$. If embedded in a larger network in such a way that the shortest paths from $\ell_1$,$\ell_2$ and $\ell_3$ to $t$ all have the same length (and no queues during the interval $[0,1]$), the inflow at node $s$ can be distributed in any way among the three edges. In particular, it is possible to send all flow over only one of the three edges. In any distribution there has to be at least one edge which carries a flow volume of at least $4$.}\label{fig:Clause-Gadget}
		\end{minipage}
	\end{figure}
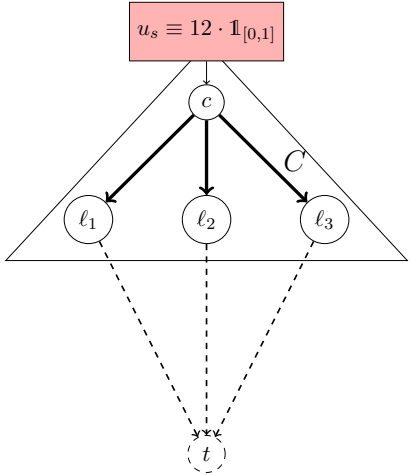
	
	We start by providing two types of gadgets: One for the clauses and one for the variables of a \ThreeSAT-formula. The clause gadget $C$ (see \Cref{fig:Clause-Gadget}) consists of a source node $c$ with a constant network inflow rate of $12$ over some interval of length $1$ and three edges with capacity $12$ and travel time $1$ connecting $c$ to the nodes $\ell_1$, $\ell_2$ and $\ell_3$, respectively. This gadget will later be embedded into a larger network in such a way that the shortest paths from the nodes $\ell_1$, $\ell_2$ and $\ell_3$ to the sink $t$ all have the same length. Thus, the flow entering the gadget at the source node $c$ can be distributed in any way over the three outgoing edges. We will have a copy of this gadget for any clause of the given \ThreeSAT-formula with the three nodes $\ell_1$, $\ell_2$ and $\ell_3$ corresponding to the three literals of the respective clause. Setting a literal to true will than correspond to sending a flow volume of at least $4$ towards the respective node.
	
	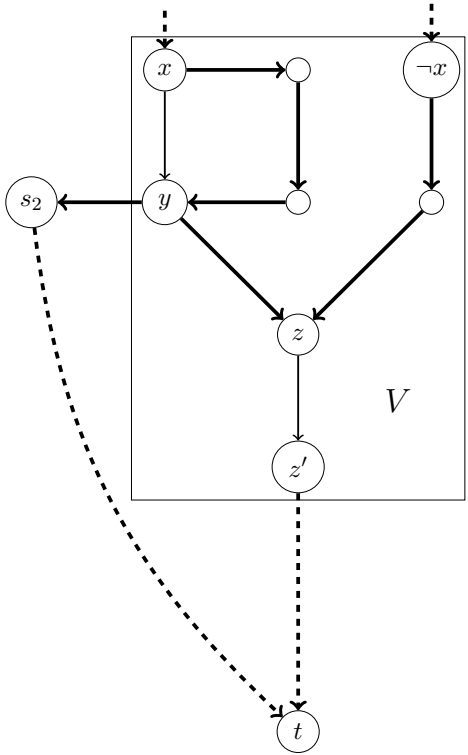
\begin{figure}\centering
		\begin{minipage}[c]{0.4\textwidth}
			\begin{adjustbox}{max width=\textwidth}
				\begin{tikzpicture}
	\node[namedVertex] (x) at (0,0) {$x$};
	\node[namedVertex] (xa) at (2,0) {};
	\node[namedVertex] (xb) at (2,-2) {};
	\node[namedVertex] (y) at (0,-2) {$y$};
	\node[namedVertex] (nx) at (4,0) {$\neg x$};
	\node[namedVertex] (ny) at (4,-2) {};
	\node[namedVertex] (z) at (2,-4) {$z$};
	\node[namedVertex] (zs) at (2,-6) {$z'$};
	\node[namedVertex] (u) at (-2,-2) {$s_2$};
	\node[namedVertex] (t) at (2,-10) {$t$};
	
	\draw (-.5,.5) rectangle (4.5,-6.5);
	\node () at (3.5,-5) {\Large $V$};
	
	\draw[edge,<-,dashed,ultra thick] (x) -- +(0,1);
	\draw[edge] (x) -- (y);
	\draw[edge,ultra thick] (x) -- (xa);
	\draw[edge,ultra thick] (xa) -- (xb);
	\draw[edge,ultra thick] (xb) -- (y);
	\draw[edge,ultra thick] (y) -- (z);
	\draw[edge,<-,dashed,ultra thick] (nx) -- +(0,1);
	\draw[edge,ultra thick] (nx) -- (ny);
	\draw[edge,ultra thick] (ny) -- (z);
	\draw[edge] (z) -- (zs);
	\draw[edge,ultra thick] (y) -- (u);
	\draw[edge,dashed,ultra thick] (zs) -- (t);
	\draw[edge,dashed,ultra thick] (u) to[bend right=20] (t);
\end{tikzpicture}
			\end{adjustbox}
		\end{minipage}\hfill
		\begin{minipage}[c]{0.58\textwidth}
			\caption{The variable gadget $V$. The edges $xy$ and $zz'$ have capacity $1$, all other edges have infinite capacity. The travel times on all (solid) edges are $1$ while the dashed lines represent paths with a length such that the travel time from $s_2$ to $t$ is the same as from $y$ over $z$ and $z'$ to $t$. If flow enters this gadget at any rate over a time interval of length one at either $x$ or $\neg x$ all flow will travel over the edge $zz'$ to the sink $t$. If, on the other hand, at both $x$ and $\neg x$ a flow of volume at least $4$ enters the gadget over an interval of length $1$ a flow volume of more than $1$ will be diverted towards $s_2$.}\label{fig:Variable-Gadget}
		\end{minipage}
	\end{figure}	
	
	The variable gadget $V$ (see \Cref{fig:Variable-Gadget}) has two nodes $x$ and $\neg x$ over which flow can enter the gadget. From both of these nodes there is a path consisting of two edges of length $1$ leading towards a common node $z$, from where another edge of length and capacity $1$ leads to node $z'$. From there the gadget will be connected to the sink node $t$ somewhere outside the gadget. The path from $\neg x$ to $z$ has infinite capacity\footnote{Throughout this construction whenever we say that an edge has ``infinite capacity'' by that we mean some arbitrary capacity high enough such that no queues will ever form on this edge. Since the network we construct will be acyclic such capacities can be constructed inductively similarly to the constant $L_e$ in the proof of \Cref{claim:BoundOnLabelSlopes}}, while the path from $x$ to $ z$ consists of one edge with capacity $1$ followed by one edge of infinite capacity with a node $y$ between the two edges. The first edge can be bypassed by a path of length $3$ and infinite capacity. From the middle node $y$ there is also a path leaving the gadget towards $t$ via some node $s_2$ outside the gadget. This path has a total length of one more than the path via $z$ and $z'$ to $t$. 
	
	We will have a copy of this gadget for every variable of the given \ThreeSAT-formula. Similarly to the clause gadget we will interpret the variable $x$ to be set to true if a flow of volume at least $4$ traverses node $x$ and the variable to be set to false if a flow volume of at least $4$ passes through node $\neg x$. If both happens at once, i.e. both $x$ and $\neg x$ each are traversed by a flow of volume at least $4$ over the span of a time interval of length $1$, we interpret this as an inconsistent setting of the variables. In this case a flow of volume more than $1$ will leave the gadget via the edge $ys_2$ during the unit length time interval three time steps later. 	
	To verify this, assume that the flow enters at nodes $x$ and $\neg x$ during $[0,1]$. Then the flow entering through $\neg x$ will start to form a queue on edge $zz'$ two time steps later. This queue will have reached a length of at least $2$ at time $3$ and, thus, still has a length of at least $1$ at time $4$. The flow entering through $x$ at first only uses edge $xy$ until a queue of length $2$ has build up there. After that, flow will only enter this edge at a rate of $1$ to keep the queue length constant, while the rest of the flow travels through the longer path towards $y$. This flow (of volume at least $1$) as well as some non-zero amount of flow from the queue on edge $xy$ will arrive at node $y$ during the interval $[3,4]$. Because of the queue on edge $zz'$ all of this flow (of volume more than $1$) will be diverted towards $s_2$.
	If, on the other hand, flow travels through only one of these two nodes over the course of an interval of length $1$  than all this flow will be forced to travel to $t$ via $z$. The third option, i.e. flow entering the gadget through both nodes but with a volume of less than $4$ at at least one of them, will not be relevant for the further proof.
	
	We can now transform a \ThreeSAT-formula into a network as follows: Take one copy of the clause gadget $C$ for every clause of the formula (each with an inflow rate of $12$ during the interval $[0,1]$ at its respective node $c$), one copy of the variable gadget $V$ for every variable and connect them in the obvious way with edges of infinite capacity and unit travel time (e.g. if the first literal of some clause is $\neg x_1$ connect the node $\ell_1$ of this clause's copy of $C$ with the node $\neg x$ of the variable $x_1$'s copy of $V$ and so on). Then add a sink node $t$ and connect the nodes $z'$ of all variable gadgets to $t$ via edges of travel time $1$ and infinite capacity. Finally, connect the node $s_2$ (which is the same for all variable gadgets) to $t$ by first an edge $s_2v$ of travel time $1$ and then another edge $vt$ of travel time $2$ and infinite capacity. The resulting network (see \Cref{fig:Reduction-Network}) has an IDE flow not using edge $s_2v$ if and only if the \ThreeSAT-formula is satisfiable: Namely, if the formula is satisfiable, take one satisfying interpretation and define a flow as follows: In every clause gadget choose one literal satisfied by the chosen interpretation and send all flow from this gadget over this literal's corresponding edge. This ensures that in the variable gadgets all flow will enter through only one of the two possible entry nodes $x$ and $\neg x$ and, as noted before, will then leave the gadget exclusively over node $z'$. If, on the other hand, the \ThreeSAT-formula is unsatisfiable every IDE flow will sent a flow volume of more than $1$ over edge $s_2v$ during the interval $[4,5]$ since in this case any flow has to have at least one variable gadget where flow volumes of at least four enter at node $x$ as well as node $\neg x$ (otherwise such a flow would correspond to a satisfying interpretation of the given \ThreeSAT-formula). This shows that the first problem stated in \Cref{thm:NPHardness} is NP-hard.
	
	\begin{figure}[ht!]\centering
		\begin{adjustbox}{width=.8\textwidth}
			\begin{tikzpicture}
	\node[namedVertex] (s2) at (0,0) {$s_2$};
	\node[namedVertex] (v) at (0,-2) {$v$};
	
	\node[namedVertex] (t) at (6,-6) {$t$};
	
	\draw[edge] (s2) -- (v);
	\draw[edge,dashed] (v) to[bend right=35] (t);
	
	\node[vertex] (V3y) at (9.4,0) {};
	\draw[edge] (V3y) to[bend left=10] (s2);
	\draw[fill=white] (5,1) rectangle (7,-3);
	\node () at (6,-1) {$V_2$};	
	
	\node[vertex] (V2y) at (5.4,0) {};
	\draw[edge] (V2y) to[bend left=10] (s2);
	\draw[fill=white] (2,1) rectangle (4,-3);
	\node () at (3,-1) {$V_1$};
	
	\node () at (8,-1) {\Large$\dots$};
	
	\draw (9,1) rectangle (11,-3);
	\node () at (10,-1) {$V_n$};
	
	\draw (3,2) -- (5,2) -- (4,3.5) -- cycle;
	\node () at (4,2.7) {$C_1$};

	\node () at (6.5,2.5) {\Large$\dots$};

	\draw (8,2) -- (10,2) -- (9,3.5) -- cycle;
	\node () at (9,2.7) {$C_k$};
	
	\node[vertex] (Cl1) at (3.5,2.3) {};
	\node[vertex] (Cl2) at (4,2.3) {};
	\node[vertex] (Cl3) at (4.5,2.3) {};
	\node[vertex] (V1x) at (3.6,0.7) {};
	\node[vertex] (V2x) at (5.4,0.7) {};
	\node[vertex] (V3nx) at (9.4,0.7) {};
	\draw[edge] (Cl1)  -- (V1x);
	\draw[edge] (Cl2)  -- (V2x){};
	\draw[edge] (Cl3)  -- (V3nx){};
	
	\node[vertex] (V1y) at (2.4,0) {};
	\draw[edge] (V1y)  -- (s2);
	
	\node[vertex] (V1z) at (3,-2.6) {};
	\node[vertex] (V2z) at (6,-2.6) {};
	\node[vertex] (V3z) at (10,-2.6) {};
	\draw[edge] (V1z) -- (t);
	\draw[edge] (V2z) -- (t);
	\draw[edge] (V3z) -- (t);
	
\end{tikzpicture}
		\end{adjustbox}
		\captionof{figure}{Schematic representation of the whole network corresponding to a \ThreeSAT-formula with clauses $C_1, \dots, C_k$ in variables $x_1, \dots, x_n$. The triangles are clause gadgets (cf. \cref{fig:Clause-Gadget}), the rectangles are variable gadgets (cf. \cref{fig:Variable-Gadget}).}\label{fig:Reduction-Network}
	\end{figure}
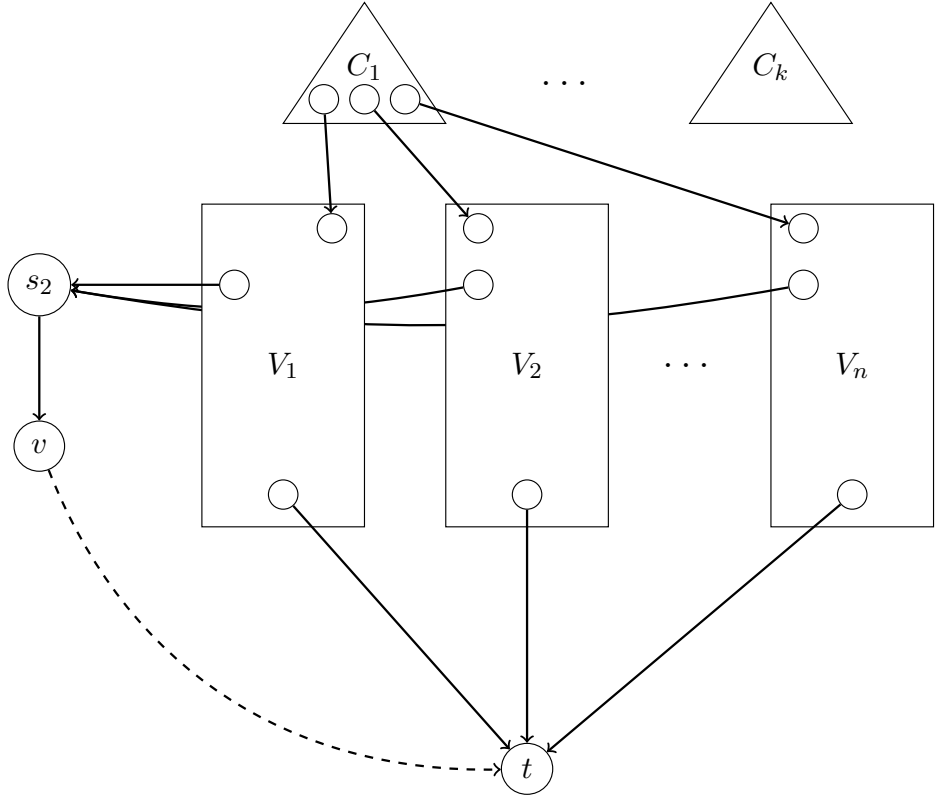
	
	In order to show that the other problems are NP-hard as well, we will introduce a third type of gadget: The indicator gadget $I$ (see \Cref{fig:Indicator-Gadget}). We can construct such a gadget for any given single-source single-sink network $\network$ with constant inflow rate over the interval $[0,\theta_0]$ at its source node. It consists of a new source node $s_1$ with the same inflow rate as $\network$'s source node shifted by $5$ time steps. The node $s_1$ is connected to the sink node $t$ (outside the gadget) by two paths: One through the network $\network$ (entering it at its original source node $s_{\network}$ and leaving it from its sink node $t_{\network}$) and one through two additional nodes $s_2$ and $v$ and an edge of capacity and travel time $1$ between them. All other edges outside $\network$ have infinite capacity. The two outgoing edge from $s_1$ both have a length of $\theta_0$. The path through the gadget has length one more than the path via $s_2$ and $v$. The node $s_2$ has a constant network inflow rate of $1$ starting at time $4$ and ending at time $5+\theta_0$. When embedding this gadget into a larger network (with sink $t$) the gadget is connected to the larger network by one or more incoming edges into $s_2$.
	
	\begin{figure}\centering
		\begin{minipage}[c]{0.4\textwidth}
			\begin{adjustbox}{max width=\textwidth}
				\begin{tikzpicture}
	\newcommand{\colComRed}{red!30}

	\node[namedVertex] (s1) at (2,0) {$s_1$};
	\draw (s1) +(0,.7) node[rectangle, draw, minimum width=1.5cm, minimum height=1cm,fill=\colComRed,anchor=south] (s1u) {$u_{s_1}\equiv u_{\network}\cdot\CharF[{[5,5+\theta_0]}]$}; 
	\draw [->] (s1u) -- (s1);
	
	\node[namedVertex] (s2) at (3,-3) {$s_2$};
	\draw (s2) +(0.7,0) node[rectangle, draw, minimum width=1.5cm, minimum height=1cm,fill=\colComRed,anchor=west] (s2u) {$u_{s_2}\equiv \CharF[{[4,5+\theta_0]}]$}; 
	\draw [->] (s2u) -- (s2);
	
	\node[namedVertex] (v) at (3,-5) {$v$};
	
	\node[namedVertex] (t) at (3,-8) {$t$};
	
	\node[namedVertex] (sN) at (1,-3) {\footnotesize $s_{\network}$};
	\node[namedVertex] (tN) at (1,-5) {\footnotesize $t_{\network}$};
	\draw (0,-2.5) rectangle (2,-5.5);
	\node () at (1,-4) {$\network$};

	\draw (-0.2,0.5) rectangle (3.5,-5.7);
	\node () at (1,-1) {\Large $I$};
	
	\draw[edge,ultra thick] (s1) -- (sN);
	\draw[edge,ultra thick] (s1) -- (s2);
	
	\draw[edge] (s2) -- (v);
	\draw[edge, <-, dashed] (s2) -- +(1,1);	
	
	\draw[edge,ultra thick] (v) -- (t);
	\draw[edge,ultra thick] (tN) -- (t);

\end{tikzpicture}
			\end{adjustbox}
		\end{minipage}\hfill
		\begin{minipage}[c]{0.59\textwidth}
			\caption{The indicator gadget $I$ for a single-source single-sink network $\network$ with network inflow rate $u_{\network} \CharF[{[0,\theta_0]}]$. All bold edges have infinite capacity, the edge $s_2v$ has capacity $1$. The edges $s_1s_{\network}$ and $s_1s_2$ both have travel time $\theta_0$, edge $s_2v$ has a travel time of $1$ and the edges $t_{\network}t$ and $vt$ can have any travel time such that the shortest $s_1$-$t$ path through $\network$ has a length of exactly one more than the $s_1$-$t$ path using edge $s_2v$. If within the interval $[4,5]$ a flow of volume more than $1$ arrives at $s_2$ over the dashed edge, all flow entering the network at $s_1$ will travel trough $\network$ (it will arrive at that sub-networks source node $s_{\network}$ at a rate of $u_{\network}$ during the interval $[5+\theta_0,5+2\theta_0]$). If, on the other hand, a flow volume of at most $1$ reaches $s_2$ via the dashed edge up to time $5+\theta_0$ all flow originating at $s_2$ will bypass $\network$ using edge $s_2v$ and $\network$ will forever remain empty.}\label{fig:Indicator-Gadget}
		\end{minipage}
	\end{figure}
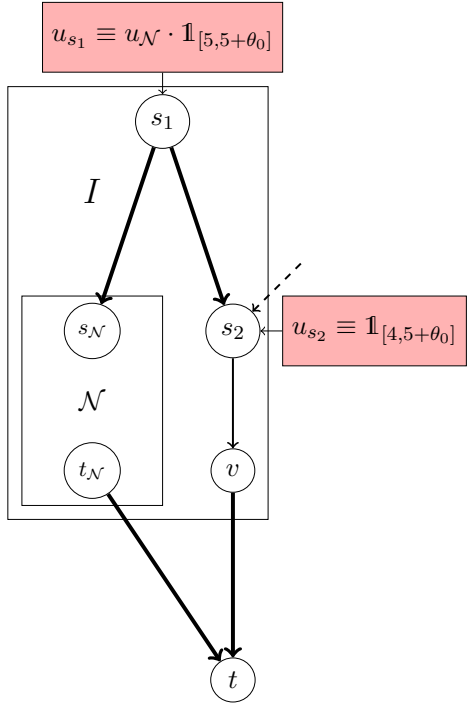
	
	If no flow ever enters the gadget via this edge, all flow generated at $s_1$ will travel through the path containing $s_2v$. If, on the other hand, a flow of volume more than $1$ comes through this edge before the inflow at node $s_1$ starts, all the flow generated there will travel through the subnetwork $\network$. Adding this gadget to the network constructed from the \ThreeSAT-formula as described above results in a network with the following properties (see \Cref{fig:3SAT-Network-Example} for an example):
	\begin{itemize}
		\item If the \ThreeSAT-formula is satisfiable there exists an IDE flow where the subnetwork $\network$ inside gadget $I$ is never used but edge $s_1s_2$ is used.
		\item If the \ThreeSAT-formula is unsatisfiable every IDE flow will be such that its restriction to the subnetwork $\network$ inside $I$ is a (time shifted) IDE flow in the original stand alone network $\network$ and the edge $s_1s_2$ is never used.
	\end{itemize}
	Accordingly, if for example we use the network from \Cref{fig:CompComplexity} as sub-network we have a reduction from \ThreeSAT{} to the fourth problem from \Cref{thm:NPHardness}. Any network $\network$ gives us a reduction to the second problem (with edge $s_1s_2$ as the special edge). And just an edge with a very small capacity allows a reduction to the third problem. Alternatively, one could also use a network wherein flow gets caught in cycles for a long time before it reaches the sink as, for example, the network constructed to prove the lower bound on the termination time of IDE in \cite{GH20PoA}.
\end{proof}

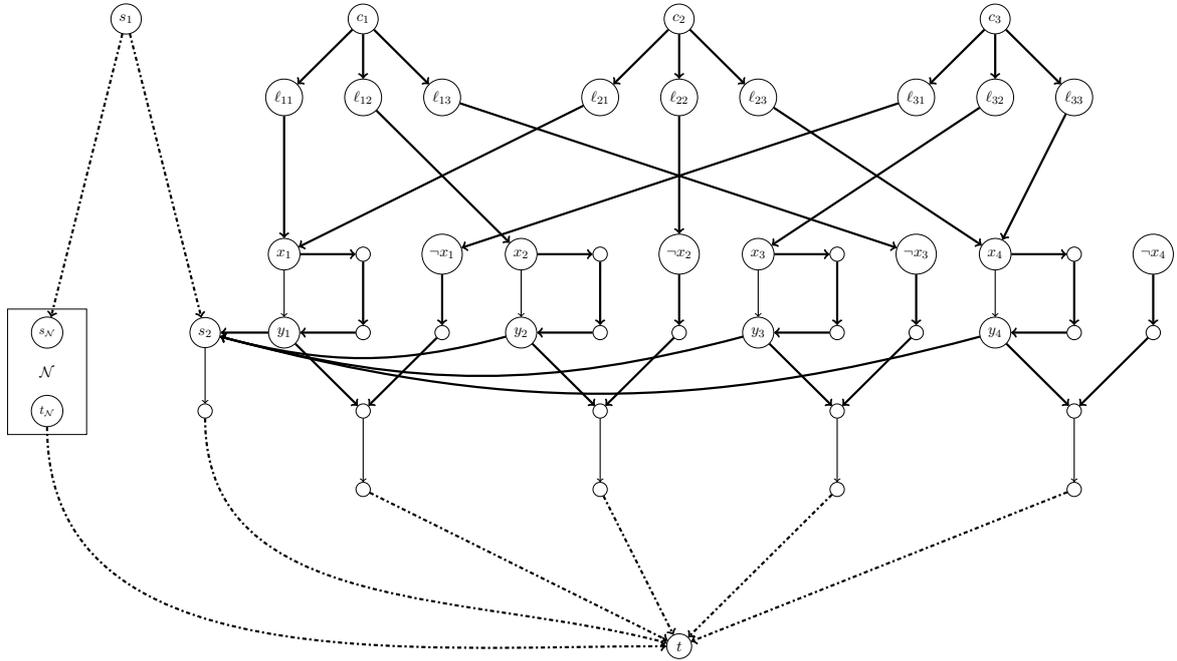
\begin{figure}[ht!]\centering
	\begin{adjustbox}{max width=\textwidth}
		\begin{tikzpicture}[scale=2]
	\node[namedVertex] (c1) at (1,0) {$c_1$};
	\node[namedVertex] (l11) at (0,-1) {$\ell_{11}$};
	\node[namedVertex] (l12) at (1,-1) {$\ell_{12}$};
	\node[namedVertex] (l13) at (2,-1) {$\ell_{13}$};
	\draw[edge,ultra thick] (c1) -- (l11);
	\draw[edge,ultra thick] (c1) -- (l12);
	\draw[edge,ultra thick] (c1) -- (l13);
	
	\node[namedVertex] (c2) at (5,0) {$c_2$};
	\node[namedVertex] (l21) at (4,-1) {$\ell_{21}$};
	\node[namedVertex] (l22) at (5,-1) {$\ell_{22}$};
	\node[namedVertex] (l23) at (6,-1) {$\ell_{23}$};
	\draw[edge,ultra thick] (c2) -- (l21);
	\draw[edge,ultra thick] (c2) -- (l22);
	\draw[edge,ultra thick] (c2) -- (l23);
	
	\node[namedVertex] (c3) at (9,0) {$c_3$};
	\node[namedVertex] (l31) at (8,-1) {$\ell_{31}$};
	\node[namedVertex] (l32) at (9,-1) {$\ell_{32}$};
	\node[namedVertex] (l33) at (10,-1) {$\ell_{33}$};
	\draw[edge,ultra thick] (c3) -- (l31);
	\draw[edge,ultra thick] (c3) -- (l32);
	\draw[edge,ultra thick] (c3) -- (l33);
	
	\node[namedVertex] (x1) at (0,-3) {$x_1$};
	\node[namedVertex] (x1a) at (1,-3) {};
	\node[namedVertex] (x1b) at (1,-4) {};
	\node[namedVertex] (y1) at (0,-4) {$y_1$};
	\node[namedVertex] (nx1) at (2,-3) {$\neg x_1$};
	\node[namedVertex] (ny1) at (2,-4) {};
	\node[namedVertex] (z1) at (1,-5) {};
	\node[namedVertex] (zs1) at (1,-6) {};
	\draw[edge,ultra thick] (x1) -- (x1a);
	\draw[edge] (x1) -- (y1);
	\draw[edge,ultra thick] (x1a) -- (x1b);
	\draw[edge,ultra thick] (x1b) -- (y1);
	\draw[edge,ultra thick] (y1) -- (z1);
	\draw[edge,ultra thick] (nx1) -- (ny1);
	\draw[edge,ultra thick] (ny1) -- (z1);
	\draw[edge] (z1) -- (zs1);

	\node[namedVertex] (x2) at (3,-3) {$x_2$};
	\node[namedVertex] (x2a) at (4,-3) {};
	\node[namedVertex] (x2b) at (4,-4) {};
	\node[namedVertex] (y2) at (3,-4) {$y_2$};
	\node[namedVertex] (nx2) at (5,-3) {$\neg x_2$};
	\node[namedVertex] (ny2) at (5,-4) {};
	\node[namedVertex] (z2) at (4,-5) {};
	\node[namedVertex] (zs2) at (4,-6) {};
	\draw[edge,ultra thick] (x2) -- (x2a);
	\draw[edge] (x2) -- (y2);
	\draw[edge,ultra thick] (x2a) -- (x2b);
	\draw[edge,ultra thick] (x2b) -- (y2);
	\draw[edge,ultra thick] (y2) -- (z2);
	\draw[edge,ultra thick] (nx2) -- (ny2);
	\draw[edge,ultra thick] (ny2) -- (z2);
	\draw[edge] (z2) -- (zs2);

	\node[namedVertex] (x3) at (6,-3) {$x_3$};
	\node[namedVertex] (x3a) at (7,-3) {};
	\node[namedVertex] (x3b) at (7,-4) {};
	\node[namedVertex] (y3) at (6,-4) {$y_3$};
	\node[namedVertex] (nx3) at (8,-3) {$\neg x_3$};
	\node[namedVertex] (ny3) at (8,-4) {};
	\node[namedVertex] (z3) at (7,-5) {};
	\node[namedVertex] (zs3) at (7,-6) {};
	\draw[edge,ultra thick] (x3) -- (x3a);
	\draw[edge] (x3) -- (y3);
	\draw[edge,ultra thick] (x3a) -- (x3b);
	\draw[edge,ultra thick] (x3b) -- (y3);
	\draw[edge,ultra thick] (y3) -- (z3);
	\draw[edge,ultra thick] (nx3) -- (ny3);
	\draw[edge,ultra thick] (ny3) -- (z3);
	\draw[edge] (z3) -- (zs3);
	
	\node[namedVertex] (x4) at (9,-3) {$x_4$};
	\node[namedVertex] (x4a) at (10,-3) {};
	\node[namedVertex] (x4b) at (10,-4) {};
	\node[namedVertex] (y4) at (9,-4) {$y_4$};
	\node[namedVertex] (nx4) at (11,-3) {$\neg x_4$};
	\node[namedVertex] (ny4) at (11,-4) {};
	\node[namedVertex] (z4) at (10,-5) {};
	\node[namedVertex] (zs4) at (10,-6) {};
	\draw[edge,ultra thick] (x4) -- (x4a);
	\draw[edge] (x4) -- (y4);
	\draw[edge,ultra thick] (x4a) -- (x4b);
	\draw[edge,ultra thick] (x4b) -- (y4);
	\draw[edge,ultra thick] (y4) -- (z4);
	\draw[edge,ultra thick] (nx4) -- (ny4);
	\draw[edge,ultra thick] (ny4) -- (z4);
	\draw[edge] (z4) -- (zs4);
	
	\node[namedVertex] (t) at (5,-8) {$t$};
	\draw[edge,ultra thick,dashdotted] (zs1) -- (t);
	\draw[edge,ultra thick,dashdotted] (zs2) -- (t);
	\draw[edge,ultra thick,dashdotted] (zs3) -- (t);
	\draw[edge,ultra thick,dashdotted] (zs4) -- (t);
			
	\node[namedVertex] (s1) at (-2,0) {$s_1$};
	\node[namedVertex] (s2) at (-1,-4) {$s_2$};
	\node[namedVertex] (v) at (-1,-5) {};
		
	\node[namedVertex] (sN) at (-3,-4) {\footnotesize $s_{\network}$};
	\node[namedVertex] (tN) at (-3,-5) {\footnotesize $t_{\network}$};
	\draw (-3.5,-3.7) rectangle (-2.5,-5.3);
	\node () at (-3,-4.5) {$\network$};
	
	\draw[edge,ultra thick,dashdotted] (s1) -- (sN);
	\draw[edge,ultra thick,dashdotted] (s1) -- (s2);
	\draw[edge] (s2) -- (v);
	\draw[edge,ultra thick,dashdotted] (v) to[out=-90,in=-195] (t);
	\draw[edge,ultra thick,dashdotted] (tN) to[out=-90,in=-180] (t);	
	
	\draw[edge,ultra thick] (y1) to (s2);
	\draw[edge,ultra thick] (y2) to[bend left=15] (s2);
	\draw[edge,ultra thick] (y3) to[bend left=15] (s2);
	\draw[edge,ultra thick] (y4) to[bend left=15] (s2);
	
	\draw[edge,ultra thick] (l11) -- (x1);
	\draw[edge,ultra thick] (l12) -- (x2);
	\draw[edge,ultra thick] (l13) -- (nx3);
	
	\draw[edge,ultra thick] (l21) -- (x1);
	\draw[edge,ultra thick] (l22) -- (nx2);
	\draw[edge,ultra thick] (l23) -- (x4);
	
	\draw[edge,ultra thick] (l31) -- (nx1);
	\draw[edge,ultra thick] (l32) -- (x3);
	\draw[edge,ultra thick] (l33) -- (x4);
\end{tikzpicture}
	\end{adjustbox}
	\caption{The whole network for the \ThreeSAT-formula $(x_1 \vee x_2 \vee \neg x_3) \wedge (x_1 \vee \neg x_2 \vee x_4) \wedge (\neg x_1 \vee x_3 \vee x_4)$. The bold edges have infinite capacity, while all other edges have capacity $1$. The solid edges have a travel time of $1$, the dashdotted edges may have variable travel time (depending on the subnetwork $\network$).}\label{fig:3SAT-Network-Example}
\end{figure}

\begin{remark}
	Combining a construction similar to the one above with the single-source multi-sink network constructed in the proof of \cite[Theorem 6.3]{GHS18} to show that multi-commodity IDE flows may cycle forever, shows that the problem to decide whether a given multi-sink network has an IDE terminating in finite time is NP-hard as well.
\end{remark}

\begin{remark}
	The above construction also shows the following aspect of IDE flows: While a network may trivially contain edges that are never used in any IDE, edges that are only used in some IDE flows and edges that are used in every IDE, there can also be edges that are either not used at all or used for some flow volume of at least $c$, but never with any flow volume strictly between $0$ and $c$.
\end{remark}


\section{Conclusions and Open Questions}

We showed that Instantaneous Dynamic Equilibria can be computed in finite time for single-sink networks by applying the natural $\alpha$-extension algorithm. The obtained explicit bounds on the required number of extension steps are quite large and we do not think that they are tight. Thus, further analysis is needed here.

We then turned to the computational complexity of IDE flows. We gave an example of a small instance which only allows for IDE flows with rather complex structure, thus, implying that the worst case output complexity of any algorithm computing IDE flows has to be exponential in the encoding size of the input instances. Furthermore, we showed that several natural decision problems involving IDE flows are NP-hard by describing a reduction from \ThreeSAT. 

One common observation that can be drawn from many proofs involving IDE flows (in this paper as well as in \cite{GHS18} and \cite{GH20PoA}) is that they often allow for some kind of \emph{local analysis} of their structure -- something which seems out of reach for Dynamic Equilibrium flows. This local argumentation allowed us to analyse the behavior of IDE flows in the rather complex instance from \Cref{sec:Hardness} by looking at the local behavior inside the much simpler gadgets from which the larger instance is constructed. At the same time, this was also a key aspect of the positive result in \Cref{sec:AlgTermination} where it allowed us to use inductive reasoning over the single nodes of the given network. We think that this local approach to the analysis of IDE flows might also help to answer further open questions about IDE flows in the future.
One such topic might be a further investigation of the computational complexity of IDE flows. While both our upper bound on the number of extension steps as well as our lower bound for the worst case computational complexity are superpolynomial bounds, the latter is at least still polynomial in the termination time of the constructed flow, which is not the case for the former. Thus, there might still be room for improvement on either bound. 

Cominetti et al.~\cite{CominettiCO17} showed (using convergence to a steady-state) that the queue length of dynamic equilibria with infinitely lasting inflow rates are bounded (provided that the inflow rate is at most the capacity of a minimal cut). As shown in \Cref{thm:NoSteadyState}, convergence to a steady state is not guaranteed for IDE flow. Therefore, we can not deduce much about the long term behavior of queue lengths in IDE flows. On the one hand, it seems intuitive that they should remain bounded as well, since, whenever queue lengths grow to high, flow particles will take a different path until the queue length has decreased again. On the other hand, the instance in \Cref{fig:CompComplexity} shows that there may exist infinitely many phases, where the inflow into the sink is smaller than the network inflow rate while at the same time there are no phases, where the inflow rate into the sink is larger than the network inflow rate. Accordingly, the queue lengths are growing higher every cycle -- although in this specific example they are still bounded by $3$.

\paragraph*{Acknowledgments:}\ 
We are grateful to the anonymous reviewers for their valuable feedback on this paper. 
Additionally, we thank the Deutsche Forschungsgemeinschaft (DFG) for their financial support.
Finally, we want to thank the organizers and participants of the 2020 Dagstuhl seminar on ``Mathematical Foundations of Dynamic Nash Flows'', where we had many helpful and inspiring discussions on the topic of this paper.

\clearpage
\bibliographystyle{plainnat}
\bibliography{literature}
\clearpage

\appendix
\section{The Waterfilling Algorithm}\label{appendix}

In order to find a possible extension of a given partial IDE at a single node we have to determine a solution to \eqref{OPT:bvThetaK}, i.e. find a distribution of the flow coming into this node to outgoing active edges in such a way that all used edges remain active for some proper time interval. As shown in \cite{GHS18} this can be done by a simple water filling procedure (\cite[Algorithm 1 (electronic supplementary material)]{GHS18}), which we will restate here for the convenience of the reader. The basic idea of this procedure is to first determine for every outgoing active edge $vw$ and all possible future constant edge inflow rates $z$ the right side derivative of the resulting shortest instantaneous travel time towards the sink for particles starting with this edge, i.e. 
\[\frac{g_e(z)}{\nu_e}+\rDeriv{\ell_{w}}(\theta),\]
where $g_e(z) \coloneqq  z-\nu_e,$ if $q_e(\theta) > 0$ and $g_e(z) \coloneqq  \max\set{z-\nu_e,0}$, otherwise. Seen as functions in $z$ these are continuous monotonic increasing functions starting with a constant part (if the edge has no queue to begin with) followed by an affine linear part. Thus, they can always be written in the form
\[k(z) = \begin{cases}
	\beta, & z \leq \gamma \\
	\beta + \frac{1}{\alpha}(z-\gamma), & z \geq \gamma
\end{cases}\]
for appropriately chosen constants $\alpha, \beta, \gamma$. The goal is now to distribute the current gross node inflow rate to the outgoing edges such that for all edges getting a non-zero part of this flow rate their respective functions $k$ evaluated at these rates coincide, while the $k$ functions of all other edges are at least as high when evaluated at a flow rate of $0$. This can be accomplished by ordering the edges with increasing value of $k$ for inflow rate $0$ and then simultaneously filling the available node inflow into the edges with currently lowest value of $k$ until all flow is distributed. This is exactly what is accomplished by \Cref{Alg:FindExtension}.

\begin{algorithm}[H] 
	\SetKwInOut{Input}{Input}\SetKwInOut{Output}{Output}
	\caption{Water filling procedure for flow distribution}\label{Alg:FindExtension}
	
	\Input{A number $b_v^-(\theta) \geq 0$ and functions $k_i: \IR_{\geq 0} \to \IR_{\geq 0}$ with $\alpha_i > 0$ for $i = 1, \dots, p \coloneqq \abs{\edgesLeaving{v} \cap E_{\theta}}$ and $\beta_1 \leq \beta_2 \leq \dots \leq \beta_{p}$.}
	\Output{Values $z_i \geq 0$ such that $\sum_{i=1}^{p} z_i = b_v^-(\theta)$ and for some $r' \leq p$ satisfying $k_0(z_0)= \dots = k_{r'}(z_{r'}) \leq \beta_{r'+1}$, $z_i>0$ for $i \leq r'$ and $z_i=0$ for $i > r'$.}
	
	\vspace{.3em}\hrule\vspace{.3em}
	
	Find the maximal $r \in \set{0,1, \dots, p}$ with $\sum_{i=1}^{r}\max\set{z|k_i(z) \leq \beta_r} \leq b_v^-(\theta)$
	
	\eIf{$r<p$ and $\sum_{i=1}^{r}\max\set{z|k_i(z) \leq \beta_{r+1}} \leq b_v^-(\theta)$}{
		
		Set $z_i \leftarrow \begin{cases}
			\max\set{z|k_i(z) \leq \beta_{r+1}}, 	&i \leq r\\
			b^-_v(\theta) - \sum_{i<r}z_i,					&i = r+1\\
			0									&i > r+1
		\end{cases}$ \label{Alg:FindExtension:DefZi1}
	}{
		Set $z_i \leftarrow \begin{cases}
			\max\set{z|k_i(z) \leq \beta_{r}}, 	&i \leq r\\
			0									&i > r
		\end{cases}$ and $b' \leftarrow b^-_v(\theta) - \sum_{i=1}^{p}z_i$ \label{Alg:FindExtension:DefZi2a}
		
		Set $z_i \leftarrow z_i + \frac{\alpha_i}{\sum_{j=1}^{r-1}\alpha_j}b'$ for all $i \leq r$. \label{Alg:FindExtension:DefZi2b}
	}
	
	\KwRet{$z_1, \dots, z_{p}$}
\end{algorithm}

The correctness of this approach has been proven in \cite[electronic supplementary material, Lemma 1]{GHS18}.

\begin{obs}\label{appobs:FlowDistributionDependence}
	The flow distribution obtained by using \Cref{Alg:FindExtension} at a given node only depends on the set of active edges, which subset of those currently has a non-zero queue, the gross node inflow rate and the label functions $\ell_{w_i}$.
\end{obs}

\begin{obs}\label{appobs:FlowDistributionRational}
	If all input data for \Cref{Alg:FindExtension} (i.e. $b_v^-(\theta)$ as well as all $\alpha_i, \beta_i$ and $\gamma_i$) is rational, so is the output (i.e. the $z_i$).
\end{obs}

\end{document}